\documentclass[english,final,table]{IEEEtran}
\usepackage[T1]{fontenc}
\usepackage[latin9]{inputenc}
\usepackage[table]{xcolor}
\usepackage{pdfcolmk}
\usepackage{array}
\usepackage{float}
\usepackage{mathrsfs}
\usepackage{bm}
\usepackage{amsmath}
\usepackage{amsthm}
\usepackage{amssymb}
\usepackage{graphicx}
\PassOptionsToPackage{normalem}{ulem}
\usepackage{ulem}

\makeatletter

\providecommand{\tabularnewline}{\\}
\floatstyle{ruled}
\newfloat{algorithm}{tbp}{loa}
\providecommand{\algorithmname}{Algorithm}
\floatname{algorithm}{\protect\algorithmname}
\providecolor{lyxadded}{rgb}{0,0,1}
\providecolor{lyxdeleted}{rgb}{1,0,0}

\DeclareRobustCommand{\lyxdeleted}[3]{{\color{lyxdeleted}\lyxsout{#3}}}
\DeclareRobustCommand{\lyxsout}[1]{\ifx\\#1\else\sout{#1}\fi}

\theoremstyle{plain}
\newtheorem{prop}{\protect\propositionname}
\theoremstyle{plain}
\newtheorem{lem}{\protect\lemmaname}

\usepackage{amsmath}
\usepackage{amssymb}
\usepackage[level=0]{wgroup_message}  
\usepackage{colortbl}
\definecolor{lightcyan}{rgb}{0.88, 1.0, 1.0}

\author{

\IEEEauthorblockN{Bowen~Li~and Junting~Chen}

\IEEEauthorblockA{School of Science and Engineering (SSE) and Shenzhen Future Network of Intelligence Institute (FNii-Shenzhen) \\ The Chinese University of Hong Kong, Shenzhen, Guangdong 518172, China}

}


\usepackage[acronym]{glossaries}
\newcommand{\newac}{\newacronym}
\newcommand{\ac}{\gls}
\newcommand{\Ac}{\Gls}
\newcommand{\acpl}{\glspl}

\makeglossaries

\newac{speb}{SPEB}{square position error bound}
\newac[plural=EFIMs,firstplural=Fisher information matrices (EFIMs)]{efim}{EFIM}{Fisher information matrix}
\newac{ne}{NE}{Nash equilibrium}
\newac{mse}{MSE}{mean squared error}
\newac{toa}{TOA}{time-of-arrival}
\newac{snr}{SNR}{signal-to-noise ratio}
\newac{lan}{LAN}{local area network}
\newac{psd}{PSD}{positive semidefinite}
\newac{pd}{PD}{positive definite}
\newac{wrt}{w.r.t.}{with respect to}
\newac{lhs}{L.H.S.}{left hand side}
\newac{wp1}{w.p.1}{with probability 1}
\newac{kkt}{KKT}{Karush-Kuhn-Tucker}
\newac{wlog}{w.l.o.g.}{without loss of generality}
\newac{mle}{MLE}{maximum likelihood estimation}
\newac{gps}{GPS}{global positioning system}
\newac{rssi}{RSSI}{received signal strength indicator}
\newac{mimo}{MIMO}{multiple-input multiple-output}
\newac{csi}{CSI}{channel state information}
\newac{fdd}{FDD}{frequency division duplexing}
\newac{ms}{MS}{mobile station}
\newac{bs}{BS}{base station}
\newac{d2d}{D2D}{device-to-device}
\newac{slnr}{SLNR}{signal-to-interference-leakage-and-noise-ratio}
\newac{ula}{ULA}{uniform linear antenna array}
\newac{pas}{PAS}{power angular spectrum}
\newac{mmse}{MMSE}{minimum mean square error}
\newac{zf}{ZF}{zero-forcing}
\newac{rzf}{RZF}{regularized zero-forcing}
\newac{as}{AS}{angular spread}
\newac{aod}{AOD}{angle of departure}
\newac{iid}{i.i.d.}{independent and identically distributed} 
\newac{sinr}{SINR}{signal-to-interference-and-noise ratio}
\newac{tdd}{TDD}{time-division duplex}
\newac{rvq}{RVQ}{random vector quantization}
\newac{rhs}{R.H.S.}{right hand side}
\newac{mrc}{MRC}{maximum ratio combining}
\newac{cdf}{CDF}{cumulative distribution function}
\newac{a.s.}{a.s.}{almost surely}
\newac{los}{LOS}{line-of-sight}
\newac{jsdm}{JSDM}{joint spatial division and multiplexing}
\newac{map}{MAP}{maximum a posteriori}
\newac{klt}{KLT}{Karhunen-Lo\`eve Transform}
\newac{lbe}{LBE}{link bargaining equilibrium}
\newac{se}{SE}{Stackelberg equilibrium}
\newac{uav}{UAV}{unmanned aerial vehicle}
\newac{nlos}{NLOS}{non-line-of-sight}
\newac{pdf}{PDF}{probability density function}
\newac{em}{EM}{expectation-maximization}
\newac{knn}{KNN}{$k$-nearest neighbor}
\newac{svd}{SVD}{singular value decomposition}
\newac{nmf}{NMF}{non-negative matrix factorization}
\newac{umf}{UMF}{unimodality-constrained matrix factorization}
\newac{rmse}{RMSE}{rooted mean squared error}
\newac{olos}{OLOS}{obstructed line-of-sight}
\newac{mmw}{mmW}{millimeter wave}
\newac{ber}{BER}{bit error rate}
\newac{rss}{RSS}{received signal strength}
\newac{lp}{LP}{linear program}
\newac{ufw}{U-FW}{unimodal Frank-Wolfe}
\newac{utf}{UTF}{unimodality-constrained tensor factorization}
\newac{fw}{FW}{Frank-Wolfe}
\newac{iot}{IoT}{Internet-of-Things}
\newac{mae}{MAE}{mean absolute error}
\newac{crb}{CRB}{Cram\'er-Rao bound}
\newac{aoa}{AoA}{angle of arrival}
\newac{wcl}{WCL}{weighted centroid localization}

\setkeys{Gin}{width=1.0\columnwidth}

\renewcommand{\lyxdeleted}[3]{{\color{lyxdeleted}{}}}

\newac{mdp}{MDP}{Markov decision process}
\newac{wsn}{WSN}{wireless sensor network}
\newac{iff}{iff}{if and only if}
\newac{qsi}{QSI}{queue state information}
\newac{ue}{UE}{user equipment}
\newac{gbs}{GBS}{ground base station}
\newac{abs}{ABS}{aerial base station}
\newac{leo}{LEO}{low earth orbit}
\newac{gcs}{GCS}{ground control station}
\newac{tdma}{TDMA}{time division multiple access}
\newac{son}{SON}{self-organized network}
\newac{ofdm}{OFDM}{orthogonal frequency division multiplexing}
\newac{uma}{UMa}{urban macro}
\newac{isac}{ISAC}{integrated sensing and communication}
\newac{sca}{SCA}{successive convex approximation}
\newac{umi}{UMi}{Urban Micro}
\newac{lr}{LR}{Lagrangian relaxation}
\newac{qam}{QAM}{quadrature amplitude modulation}
\newac{tsp}{TSP}{travelling salesman problem}
\newac{qos}{QoS}{quality of service}
\newac{np}{NP}{non-deterministic polynomial-time}
\newac{bcd}{BCD}{block coordinate descent}

\usepackage[noadjust]{cite}
\usepackage{color}  
\usepackage{graphicx,psfrag,cite,subfigure}

\makeatother

\usepackage{babel}
\providecommand{\lemmaname}{Lemma}
\providecommand{\propositionname}{Proposition}

\begin{document}
\title{Radio Map Assisted Routing and Predictive Resource Allocation over
Dynamic Low Altitude Networks}
\maketitle
\begin{abstract}
Dynamic low altitude networks offer significant potential for efficient
and reliable data transport via \acpl{uav} relays which usually operate
with predetermined trajectories. However, it is challenging to optimize
the data routing and resource allocation due to the time-varying topology
and the need to control interference with terrestrial systems. Traditional
schemes rely on time-expanded graphs with uniform and fine time subdivisions,
making them impractical for interference-aware applications. This
paper develops a dynamic space-time graph model with a cross-layer
optimization framework that converts a joint routing and predictive
resource allocation problem into a joint bottleneck path planning
and resource allocation problem. We develop explicit deterministic
bounds to handle the channel uncertainty and prove a monotonicity
property in the problem structure that enables us to efficiently reach
the globally optimal solution to the predictive resource allocation
subproblem. Then, this approach is extended to multi-commodity transmission
tasks through time-frequency allocation, and a bisection search algorithm
is developed to find the optimum solution by leveraging the monotonicity
of the feasible set family. Simulations verify that the single commodity
algorithm approaches global optimality with more than 30 dB performance
gain over the classical graph-based methods for delay-sensitive and
large data transportation. At the same time, the multi-commodity method
achieves 100X improvements in dense service scenarios and enables
an additional 20 dB performance gain by data segmenting.
\end{abstract}


\section{Introduction\label{sec:intro}}

\mysubsubnote{Adust the introdcution according to the following outline.}

There have been rapidly growing human and autonomous robot activities
in the low altitude airspace, which typically extends from the rooftop
to 1,000 meters above ground, where the operations may include the
transport of commodities and cargo, short-distance transit and tourism,
and emergency operations \cite{WuXuZenNg:J21,ZhoSheLiHan:J23}. The
\ac{uav} network forms a wireless communication network that is capable
of operating parallel to the terrestrial cellular network \cite{BaiZhaZhaCha:J23,FanWuXia:J24}.
The key characteristic of the \ac{uav} network is that the network
topology is time-varying, but usually \emph{predictable} according
to the primary missions of the \acpl{uav}, such as pre-planned cargo
delivery following registered flight paths \cite{XiaZhoDaiQu:J23,LiChe:J24}.
Such a time-varying but predictable network topology induces a new
degree of freedom for network communication, where one can optimize
for large-timescale opportunistic transmission for delay-tolerant
data transportation.

Specifically, we consider to exploit the low altitude \ac{uav} network
for delay-tolerant data transportation, where the data package is
routed over the dynamic \ac{uav} network, and the data source and
sink can be sensor nodes or fusion units located on the ground or
in the airspace. Such a network transmission model finds many applications,
such as in environment monitoring and content distribution for caching.
However, transmission over the low altitude \ac{uav} network imposes
the following challenges:
\begin{itemize}
\item \emph{Dynamical topology}: The topology constructed by \acpl{uav}
changes over time, leading to the \ac{csi} between the transmitting
and receiving nodes (whether intended or not) varies over time.
\item \emph{Air-to-ground interference}: Aerial transmissions may cause
substantial interference with terrestrial cellular networks, due to
the high likelihood of \ac{los} conditions from \acpl{uav} in the
sky.
\end{itemize}

Routing protocols in dynamic networks are typically classified into
social-aware and mobility-aware routing based on their ability to
predict future network topology. In the absence of accurate topology
information, most methods focus on estimating transmission likelihood
by analyzing social properties and selecting routes accordingly \cite{HanHuiKumMar:J12,GanJai:J23,BabP:J24}.
However, these methods often fail to guarantee \ac{qos} due to inaccuracies
in the prediction mechanisms. In addition, these methods do not capitalize
on the \emph{predictive} property of aerial networks, where \ac{uav}
trajectories are predetermined before tasks like cargo delivery commence,
thus allowing for predicting information about future network topology.
When future topology information is available, digital twins, radio
maps, or channel models can predict future \ac{csi}. This facilitates
the use of graph-based routing protocols for robust service \cite{LiLuXueZha:J19,LiWanJinChe:J14,QuDonDaiWei:J19,LiuZhuYanLi:J22,LiCheHuaYin:J15,JiaZhaYanYua:J19,HanXuZhaWan:J23}.
Some approaches partition the dynamic network into static graph snapshots
and select transmission paths for each snapshot \cite{LiLuXueZha:J19}.
To build the temporal relationship over snapshots, time-expanded graphs
and time-space-combined routing algorithms have been introduced. For
instance, some methods \cite{LiWanJinChe:J14,QuDonDaiWei:J19,LiuZhuYanLi:J22}
enable multi-hop transmissions within a single time slot by subdividing
time into smaller intervals, ensuring stable channel gains and avoiding
causality issues, though this complicates large timescale optimization.
Other methods \cite{LiCheHuaYin:J15,JiaZhaYanYua:J19,HanXuZhaWan:J23}
restrict each time slot to a single-hop transmission, avoiding causality
concerns. However, classical uniform slot durations are inadequate
for low altitude aerial networks, which must adapt their transmission
strategies based on terrestrial network conditions, and determining
the optimal slot length remains challenging.

Some recent works \cite{MeiZha:J21,VaeLinZha:J24,MeiWuZha:J19,LiuZhe:J24,HuaMeiXu:J19,MeiZha:J20,HouDenShi:J21}
attempted to mitigate air-to-ground interference by time/frequency
orthogonalization, beam orthogonalization, and path design. For example,
works \cite{MeiWuZha:J19,LiuZhe:J24} focus on allocating spectrum
resources to aerial nodes based on real-time cellular network demands,
effectively separating aerial and terrestrial networks temporally
and spectrally. Works \cite{MeiZha:J20,HouDenShi:J21} employ directional
beamforming that targets specific receivers and reduces interference
to adjacent nodes by controlling the side lobes. However, these methods
require small-scale \ac{csi}, which is costly and sometimes unfeasible,
because small-scale \ac{csi} may not be predictable due to the randomness
nature of wireless channels. Works \cite{HuaMeiXu:J19,LiuZhe:J24}
optimize \ac{uav} positions to keep them away from ground users,
thus minimizing interference. Thus, there methods are not applicable
to the case where \ac{uav} trajectories are determined and cannot
be altered. Some preliminary results reported in \cite{LiChe:J24,LiChe:J24b}
have shown that it is possible to exploit the large-scale \ac{csi}
for predictive transmission optimization with air-to-ground interference
control when the routing is determined as fixed. However, when the
routes are to be optimized, the route selection and the transmission
timing optimization are coupled, whereas a brute-force search for
the best route requires exponential complexity.

In this paper, we study the interference-aware predictive communications
in low altitude dynamic aerial networks. With the aid of radio maps,
a predictive problem on the large timescale for route, transmission
power, and timing (time boundary) planning is formulated. Towards
this end, two main technical challenges are needed to be addressed:
\begin{itemize}
\item How to jointly optimize the route and the transmission timing in a
network with dynamic topology.
\item How to minimize air-to-ground interference when ensuring aerial communication
quality.
\end{itemize}
$\quad$To tackle these challenges, we develop a dynamic space-time
graph model with virtual edges and formulate a cross-layer interference-aware
optimization problem. Consequently, the routing optimization can be
solved using a bottleneck path planning algorithm, while the power
and timing allocation optimization is reformulated as an inner-outer
problem. Our key contributions are made as follows:
\begin{itemize}
\item We propose a dynamic space-time graph model with an algorithm framework
for the joint routing and predictive resource allocation over a dynamic
network.
\item We develop explicit deterministic bounds to handle the channel uncertainty
for an efficient cross-layer algorithm design. We also prove a monotonicity
property in the problem structure that enables us to efficiently reach
the globally optimal solution to the predictive resource allocation
subproblem.
\item For multi-commodity transportation, we decouple the problem to multiple
parallel single commodity transportation subproblems. While the subproblems
are coupled via the shared time-frequency resources, we show that
the optimal time-frequency allocation can be efficiently found via
exploiting a monotonicity property of the problem formulation.
\item Simulations show that the single commodity algorithm achieves near-global
optimality, providing 30 dB improvements for delay-sensitive and large
data transportation compared to classical graph-based algorithms.
Additionally, the multi-commodity algorithm delivers a 100X improvement
in dense service scenarios. For a single large commodity, segmenting
it into smaller parts for transmission further achieves an additional
20 dB performance improvement.
\end{itemize}
\ \ \ The rest of the paper is organized as follows. Section \ref{sec:System-Model}
presents the communication system model, the graph model, and the
problem formulation. Section \ref{sec:alg_single} develops the single
commodity transportation strategy based on the dynamic space-time
graph with virtual edges. Section \ref{sec:System-Model} extends
the single commodity strategy to the multi-commodity transportation.
Numerical results are demonstrated in Section \ref{sec:Simulation},
and conclusions are given in Section \ref{sec:Conclusion}.

\section{System Model\label{sec:System-Model}}

\begin{figure}
\begin{centering}
\includegraphics[width=0.9\columnwidth]{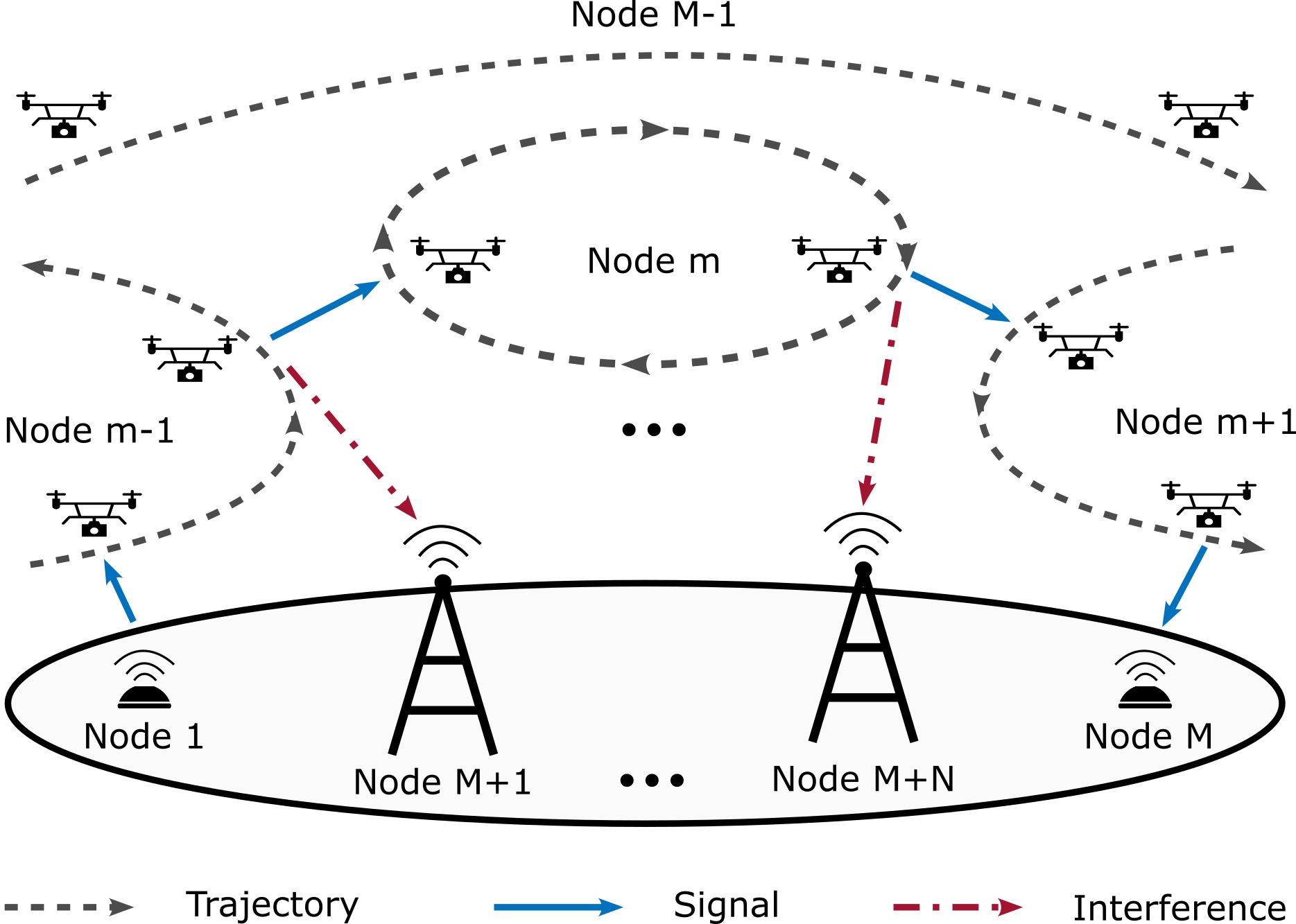}
\par\end{centering}
\caption{\label{fig:system_model}Interference-aware aerial communication system
model. The task is to transport data from node $1$ to node $M$ via
\acpl{uav} indexed from $2$ to $M-1$, while minimizing interference
to neighboring nodes indexed from $M+1$ to $M+N$.}
\end{figure}

\begin{center}
\begin{table}
\caption{\label{tab:key_notations}Key Notations}

\centering{}%
\begin{tabular}{>{\raggedright}p{0.26\columnwidth}|p{0.64\linewidth}}
\hline 
\textbf{Symbols} & \textbf{Meaning}\tabularnewline
\hline 
\rowcolor{lightcyan}$\mathcal{M}$, $\mathcal{N}$ & Node sets of the aerial and neighbor networks (Section \ref{sec:System-Model}).\tabularnewline
$S$ & Size of the data package (Section \ref{sec:System-Model}).\tabularnewline
\rowcolor{lightcyan}$h_{m,n}$, $g_{m,n}$, $\xi_{m,n}$ & Instantaneous channel gain, average channel gain, and small-scale
fading between node $m$ and node $n$ (Section \ref{subsec:channel_model}).\tabularnewline
$p_{m,n}$, $c_{m,n}$ & Instantaneous power policy and channel capacity from node $m$ to
$n$ (Section \ref{subsec:channel_model}).\tabularnewline
\rowcolor{lightcyan}$I_{m,j}$ & Instantaneous air-to-ground interference from aerial node $m$ to
ground node $j$ (Section \ref{subsec:channel_model}).\tabularnewline
$w_{m,n}^{k}(t_{k},t_{k+1})$ & Minimum worst-case interference during the transmission from $m$
to $n$ over $[t_{k},t_{k+1})$ (Section \ref{subsec:comm_model}).\tabularnewline
\rowcolor{lightcyan}$t_{k}$, $o(k)$ & The $k$th time boundary and selected relay (Section \ref{subsec:graph_model}).\tabularnewline
$\begin{aligned} & \mathscr{G}(\mathbf{t})=\\
 & \quad(\hat{\mathcal{M}},\hat{\mathcal{E}},\hat{\mathcal{W}}(\mathbf{t}))
\end{aligned}
$ & Dynamic space-time graph with node, edge, and weight sets (Section
\ref{subsec:graph_model}).\tabularnewline
\rowcolor{lightcyan}$\mathcal{Z}$ & Set of commodities (Section \ref{subsec:multi_commodity_model}).\tabularnewline
$s_{z}$, $d_{z}$ & Source and destination of commodity $z$ (Section \ref{subsec:multi_commodity_model}).\tabularnewline
\rowcolor{lightcyan}$t_{k,z}$, $o(k,z)$ & $k$th time boundary and relay for commodity $z$ (Section \ref{subsec:multi_commodity_model}).\tabularnewline
$l_{z}(t)$ & Instantaneous normalized time-frequency resource for commodity $z$
(Section \ref{subsec:multi_commodity_model}).\tabularnewline
\rowcolor{lightcyan}$\max_{k\in\mathcal{M}}\{w^{k}\}$ & Maximum value among a set of known scalars $w^{k}$; no optimization
is involved.\tabularnewline
$\max_{p\in\mathcal{P}}f(p)$ & Maximization over decision variable $p$; represents an optimization
problem.\tabularnewline
\hline 
\end{tabular}
\end{table}
\par\end{center}

Consider a data transportation task that delivers a data package of
size $S$ from a source node to a destination node via a group of
$M-2$ aerial nodes, as shown in Fig.~\ref{fig:system_model}. For
the ease of exposition, denote the source node as the $1$st node
and the destination node as the $M$th node, and thus, the set $\mathcal{M}=\{1,2,\dots,M\}$
of all $M$ nodes forms an aerial communication network. Besides,
there is a set $\mathcal{N}=\{M+1,M+2,\dots,M+N\}$ of nodes in the
neighbor network that requires interference protection. The positions
or trajectories of these nodes, defined as $\mathbf{q}_{m}(t)\in\mathbb{R}^{3}$,
$m\in\mathcal{M}\cup\mathcal{N}$, are known for a certain time horizon
$t\in[0,T]$. The key notations used in the paper are concluded in
Table \ref{tab:key_notations}.

\subsection{Channel Model and Radio Model\label{subsec:channel_model}}

We consider a flat fading channel model, where the instantaneous channel
power gain between two nodes $m\neq n\in\mathcal{M}\cup\mathcal{N}$
is given by 
\begin{equation}
h_{m,n}\left(t\right)=g_{m,n}\left(t\right)\xi_{m,n}\left(t\right),\,t\in\left[0,T\right]\label{eq:channel_model}
\end{equation}
where $g_{m,n}(t)$ is the expected channel gain and $\xi_{m,n}(t)$
is a random variable following $\text{Gamma}(\kappa_{m,n}(t),1/\kappa_{m,n}(t))$
distribution to capture the small-scale fading. Accordingly, the power
gain at time $t$ follows Gamma distribution with $\text{Gamma}(\kappa_{m,n}(t),g_{m,n}(t)/\kappa_{m,n}(t))$.

Assume the large-scale channel statistics between any two positions
$\mathbf{q}_{m}(t)$ and $\mathbf{q}_{n}(t)$ are known in advance
and captured by a predefined function 
\[
(g_{m,n}(t),\kappa_{m,n}(t))=\Xi(\mathbf{q}_{m}(t),\mathbf{q}_{n}(t)).
\]
This assumption is feasible with the use of radio maps, which are
data-driven models that correlate the locations of transmitters and
receivers to large-scale \ac{csi}, including path loss, shadowing,
and the statistics of small-scale fading \cite{SatSutIna:J21,liuche:J23,ZenCheXuWu:J24}.
By integrating the predictive positions of the nodes, \emph{e.g.},
$\mathbf{q}_{m}(t)$ and $\mathbf{q}_{n}(t)$, it becomes possible
to forecast large-scale channel gains $g_{m,n}(t)$ and the statistics
of small-scale fading $\xi_{m,n}(t)$ over time. Consequently, one
can predict the channel power gain distribution over time and effectively
plan communication strategies in advance. In other words, terrestrial
control center is able to access the required coefficients without
\ac{uav} real-time feedback. Note that $h_{m,n}(t)$ is the instantaneous
power gain at time $t$ that is not available ahead of time due to
the randomness of $\xi_{m,n}\left(t\right)$ from the small-scale
fading.\footnote{The proposed method remains effective when considering radio map construction
errors, as the instantaneous channel gain can still be approximated
by a Gamma distribution, by using the method in \cite{AlYan:J10}.}

On the other hand, the coverage radio maps of the neighboring network
nodes, denoted by $\Xi_{j}$ for $j\in\mathcal{N}$, map an aerial
position $\mathbf{q}_{n}(t)$ to the corresponding power gain from
neighbor node $j$. These radio maps characterize the received power
gain at position $\mathbf{q}_{n}(t)$ from neighboring node $j$,
and are utilized to predict the ground-to-air interference experienced
by the aerial nodes. Mathematically, the ground-to-air interference
from node $j$ to aerial node $n$ at time $t$ can be
\[
I_{nj}(t)=\Xi_{j}(\mathbf{q}_{n}(t)).
\]

Denote the transmission power of node $m$ targeted to node $n$ as
$p_{m,n}(t)\ge0$. Then, the received \ac{sinr} for node $n$ is
$p_{m,n}(t)h_{m,n}(t)/\delta_{n}^{2}(t)$ where $\delta_{n}^{2}(t)\triangleq\delta^{2}+\sum_{j}I_{nj}(t)$
and $\delta^{2}$ is noise power. To simplify the notation, we assume
that the interference-plus-noise term $\delta_{n}^{2}(t)$ is constant
across all nodes and time, and denote it simply as $\delta^{2}$.
It is worth noting that this assumption is made purely for notational
convenience and the proposed method remains valid and applicable even
when $\delta_{n}^{2}(t)$ varies across nodes and over time.

Assuming perfect Doppler compensation through advanced techniques
\cite{GuoZhaMuGao:J19,GonLiJiaWin:J21,LuZen:J24}, the \emph{instantaneous}
capacity from node $m$ to node $n$ is modeled as
\begin{equation}
c_{m,n}\left(t\right)=B\log_{2}\left(1+p_{m,n}\left(t\right)h_{m,n}\left(t\right)/\delta^{2}\right)\label{eq:def_c}
\end{equation}
where $B$ is the transmission bandwidth.

Meanwhile, the transmission from node $m$ will generate interference
to the neighbor nodes $j\in\mathcal{N}$, and the \emph{instantaneous}
interference power is modeled by
\begin{equation}
I_{m,j}\left(t\right)=\sum_{n\in\mathcal{M}}p_{m,n}\left(t\right)h_{m,j}\left(t\right).\label{eq:intf_l_p}
\end{equation}

\subsection{Link-Level Communication Model\label{subsec:comm_model}}

The key challenge of wireless communication in a dynamic network is
that there may not exist an \emph{instantaneous} end-to-end route
from the source to the destination with a satisfactory communication
quality, because the instantaneous end-to-end communication quality
for a multi-hop channel is determined by the capacity of the worst
link. For example, the destination node may be temporarily isolated
from all the other nodes, and thus, no end-to-end communication can
be established. However, some links among other communication nodes
in $\mathcal{M}$ may still experience good channel quality during
this period. As a result, the nodes have to temporarily cache the
data and pass it forward when the communication quality is good.

Specifically, we adopt a \emph{cache-and-pass} communication strategy,
where the entire data package of size $S$ is transferred completely
from one node to the other node before it is forwarded to the third
node. As a result, at each hop, only one node is selected as the target
for transporting data package, and the instantaneous interference
power in (\ref{eq:intf_l_p}) becomes
\begin{equation}
I_{m,j}\left(t\right)=p_{m,n}\left(t\right)h_{m,j}\left(t\right).\label{eq:intf_eqv}
\end{equation}
The timing of the transportation of the data package and the route
from the source to the destination is to be jointly optimized in this
paper.

Suppose that, during the allocated time interval $[t_{k},t_{k+1})$
for the $k$th hop, a node $m\in\mathcal{M}$ is scheduled to transport
the entire data package during this time segment to a node $n\in\mathcal{M}$.
Define $p_{m,n}(t)$ as the power allocation {\em policy}, which
maps the instantaneous channel state to the transmit power used to
deliver the data package from node $m$ to node $n$. Then, according
to the capacity definition in (\ref{eq:def_c}), the expected throughput
between node $m$ and node $n$ over the interval $[t_{k},t_{k+1})$
is 
\begin{equation}
\int_{t_{k}}^{t_{k+1}}\mathbb{E}\left[B\log_{2}\left(1+p_{m,n}\left(t\right)h_{m,n}\left(t\right)/\delta^{2}\right)\right]\label{eq:def_thp}
\end{equation}
and the worst-case interference to the neighboring network during
transmission is
\begin{equation}
\vartheta_{m,n}=\max_{t\in[t_{k},t_{k+1})}\left\{ \max_{j\in\mathcal{N}}\left\{ p_{m,n}\left(t\right)h_{m,j}\left(t\right)\right\} \right\} \label{eq:def_worst_case_itf}
\end{equation}
based on the instantaneous interference, defined in (\ref{eq:intf_eqv}).

We aim to minimize the worst-case interference, subject to the constraint
that the entire data of size $S$ is delivered. Denote the $w_{m,n}^{k}(t_{k},t_{k+1})$
as the minimum achievable worst-case interference, which serves as
the objective value of the following optimization problem\footnote{Note that one may also consider minimizing the total transmission
power under maximum power or interference constraints, as in the conference
version of this paper \cite{LiChens:C24}, or formulate the corresponding
dual problems. A similar solution approach may apply in such cases.} 
\begin{align}
\mathscr{P}\text{1: }\underset{\left\{ p_{m,n}\left(t\right)\right\} ,\vartheta_{m,n}}{\text{minimize}} & \ \vartheta_{m,n}\label{eq:obj_all_itf}\\
\text{subject to } & \ \int_{t_{k}}^{t_{k+1}}\mathbb{E}[B\log_{2}(1+p_{m,n}(t)/\delta^{2}\nonumber \\
 & \ \quad\quad\quad\quad\quad\quad\times h_{m,n}(t))]dt\ge S\label{eq:c_thp}\\
 & \ p_{m,n}\left(t\right)h_{m,j}\left(t\right)\le\vartheta_{m,n}\nonumber \\
 & \ \quad\quad\quad\quad\quad\forall j\in\mathcal{N},t\in\left[t_{k},t_{k+1}\right)\label{eq:c_itf}
\end{align}
where constraint (\ref{eq:c_itf}) represents the epigraph form of
the worst-case interference defined in (\ref{eq:def_worst_case_itf})
.

It may also be possible to slice one package into two, with size $S_{1}$
and $S_{2}$, $S_{1}+S_{2}=S$, and optimize two routes for transferring
the two packages. This becomes a multi-commodity transportation problem,
which will be discussed in Section \ref{sec:multi_task} as an extension
of the single commodity transportation problem as we focus here.

\subsection{Dynamic Space-Time Graph with Virtual Edges\label{subsec:graph_model}}

To jointly optimize the routing and predictive allocation over a dynamic
network, we resort to a graph-based approach and develop a dynamic
space-time graph with virtual edges.

\subsubsection{Graph model}

The aerial wireless communication network is modeled as a dynamic
space-time graph, denoted as $\mathscr{G}(\mathbf{t})=(\hat{\mathcal{M}},\hat{\mathcal{E}},\hat{\mathcal{W}}(\mathbf{t}))$,
under the allocated time boundaries $\mathbf{t}\triangleq(t_{1},t_{2},\cdots,t_{M})$,
as shown in Fig.~\ref{fig:dts_graph}. Here, $\hat{\mathcal{M}}=\{\mathcal{M}_{1},\mathcal{M}_{2},\dots,\mathcal{M}_{M}\}$
is a collection of node layers. Each layer $\mathcal{M}_{k}$ includes
all nodes in $\mathcal{M}$, representing a network snapshot at time
$t_{k}$. $\hat{\mathcal{E}}=\{\mathcal{E}_{1},\mathcal{E}_{2},\cdots,\mathcal{E}_{M-1}\}$
is a collection of directed edge sets. Each $\mathcal{E}_{k}=\{(m,n)\}_{m\in\mathcal{M}_{k},n\in\mathcal{M}_{k+1}}$contains
the directed edges from layer $\mathcal{M}_{k}$ to layer $\mathcal{M}_{k+1}$,
where the edges $(m,n)$ are strictly allowed only from a lower-layer
node $m\in\mathcal{M}_{k}$ to an upper-layer node $n\in\mathcal{M}_{k+1}$,
representing data flows from node $m$ to node $n$ during the interval
$[t_{k},t_{k+1})$. In addition, $\hat{\mathcal{W}}(\mathbf{t})=\{\mathbf{W}_{1}(\mathbf{t}),\mathbf{W}_{2}(\mathbf{t}),\mathbf{W}_{M-1}(\mathbf{t})\}$
is a collection of the weight matrices. Each matrix $\mathbf{W}_{k}(\mathbf{t})=\{w_{m,n}^{k}(t_{k},t_{k+1})\}_{(m,n)\in\mathcal{E}_{k}}$
includes the weight for the edges in $\mathcal{E}_{k}$, where $w_{m,n}^{k}(t_{k},t_{k+1})$
denotes the weight of edge $(m,n)$, representing the interference
cost incurred during the data transfer from node $m$ to node $n$
over the interval $[t_{k},t_{k+1})$.

\begin{figure}
\begin{centering}
\includegraphics{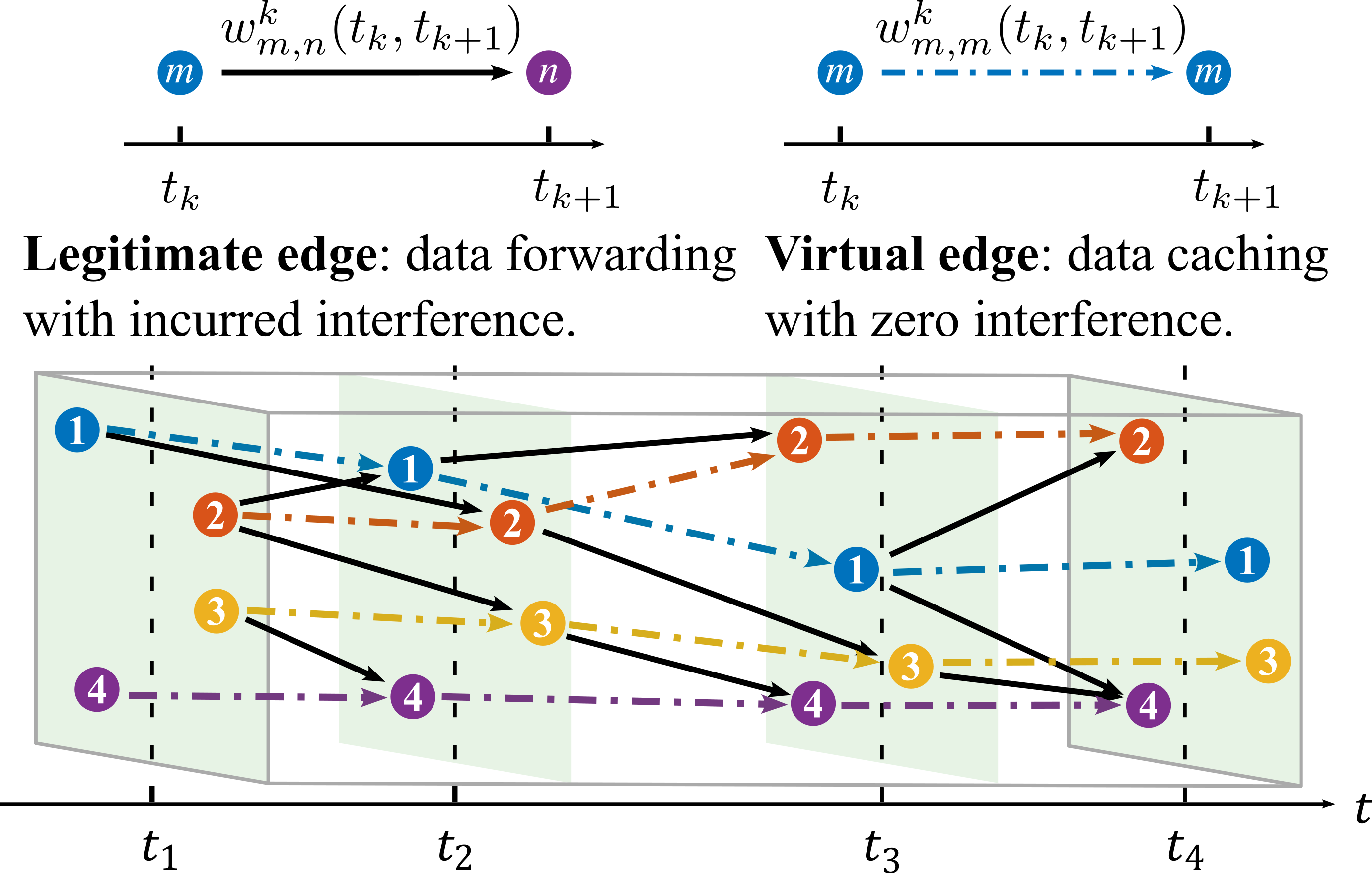}
\par\end{centering}
\caption{\label{fig:dts_graph}Dynamic space-time graph model with virtual
edges. Distinct nodes in adjacent layers are connected by legitimate
edges with weight $w_{m,n}^{k}(t_{k},t_{k+1})$, shown as solid lines.
Identical nodes in adjacent layers are connected by virtual edges
with zero weight, shown as dotted lines.}
\end{figure}

\subsubsection{Physical meaning}

Physically, edges $(m,n)$ with $m\neq n$, referred to as \emph{legitimate
edges} (solid lines in Fig.~\ref{fig:dts_graph}), represent forwarding
during $[t_{k},t_{k+1})$, whereas edges with $m=n$, referred to
as \emph{virtual edges} (dotted lines in Fig.~\ref{fig:dts_graph}),
represent caching over $[t_{k},t_{k+1})$. For legitimate edges $(m,n)$
with $m\neq n$, the weight $w_{m,n}^{k}(t_{k},t_{k+1})$ corresponds
to the optimal value obtained from solving Problem $\mathscr{P}1$.
In contrast, for virtual edges $(m,n)$ with $m=n$, the weight $w_{m,n}^{k}(t_{k},t_{k+1})$
is set to zero, as no actual transmission or interference occurs.

As a result, the path planning problem for multi-hop data transportation
corresponds to selecting a path that connects node 1 in $\mathcal{M}_{1}$
to node $M$ in $\mathcal{M}_{M}$ in the dynamic space-time graph
$\mathscr{G}(\mathbf{t})$. The optimal path is defined as the one
with the lowest cost, where the cost is given by
\begin{equation}
\vartheta\triangleq\max_{k\in\left\{ 1,2,\cdots,M-1\right\} }\left\{ w_{o\left(k\right),o\left(k+1\right)}^{k}\left(t_{k},t_{k+1}\right)\right\} \label{eq:def_net_itf}
\end{equation}
and $o\left(k\right)$ denotes the selected node (or relay) in the
$k$th layer along the path.

Note that the legitimate-edge weight $w_{m,n}^{k}(t_{k},t_{k+1})$,
as defined in Problem~$\mathscr{P}1$, captures the impact of node
velocity, mobility, and channel quality. As a result, the network-level
solution obtained via the dynamic space-time graph naturally accounts
for these essential physical-layer characteristics.

\subsubsection{Advantages}

The proposed dynamic space-time graph, by fixing the temporal dimension
$M$, reduces the optimization complexity. Furthermore, the incorporation
of virtual edges ensures that this simplification does not compromise
the solution's optimality. 
\begin{figure}
\begin{centering}
\includegraphics{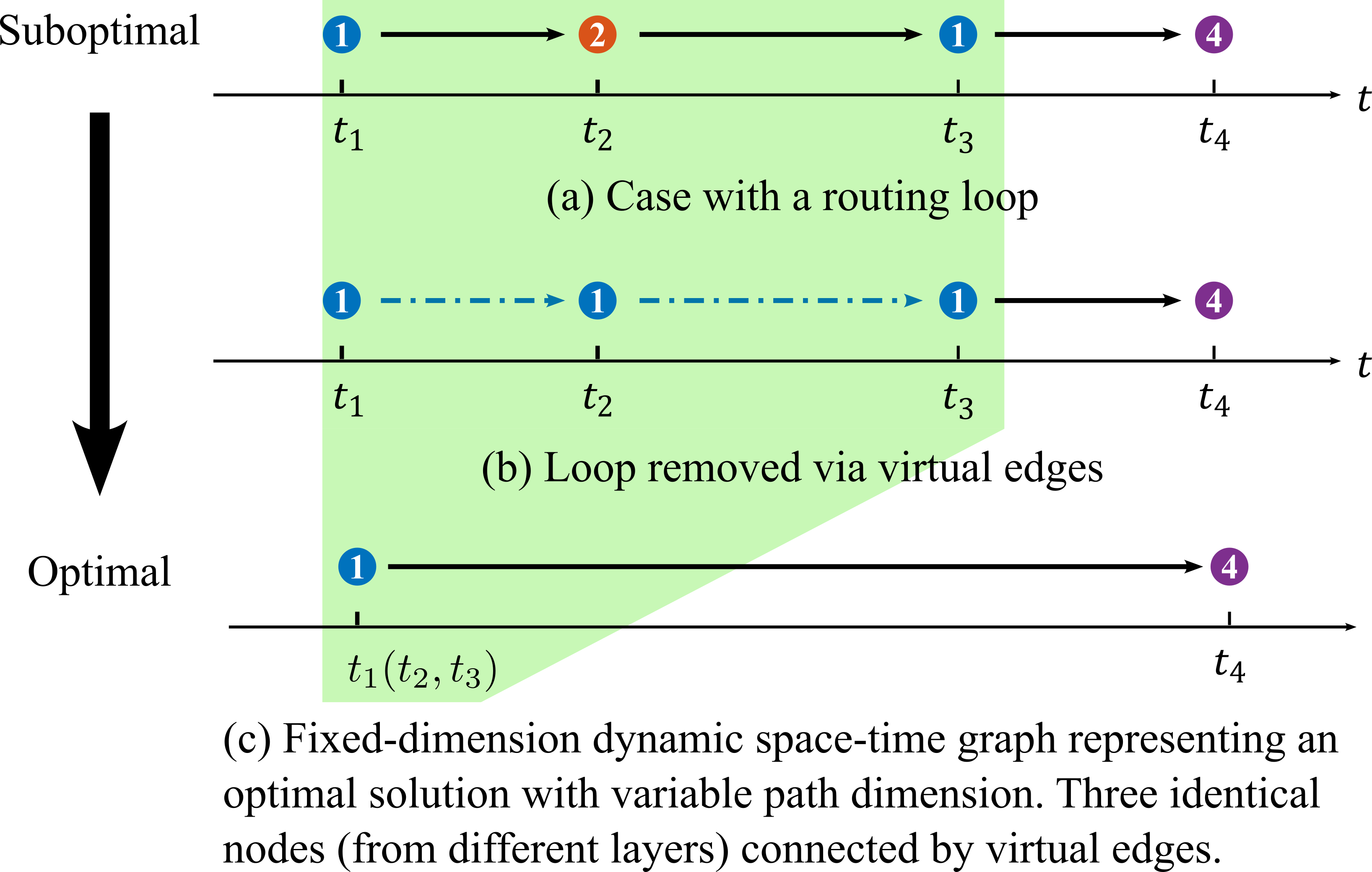}
\par\end{centering}
\caption{\label{fig:case_example_route}Progressive improvement of routing
using virtual edges in the dynamic space-time graph. Suppose the optimal
solution is the path $1\to4$ with time boundary $\mathbf{t}^{*}=(t_{1},t_{4})$.
The three cases illustrate how loop elimination and the use of virtual
edges enable a transition from a suboptimal to an optimal path under
a fixed graph dimension. Case (b) improves upon (a) by eliminating
the routing loop through cost-free virtual edges. Case (c) further
enhances performance by allocating more available time to reduce interference
and achieve the minimum-cost path.}
\end{figure}

First, unlike conventional space-time graph models that require joint
optimization over both the time boundary vector and its dimension,
our approach reduces the search space by restricting the problem to
optimizing only the individual time values $t_{k}$. Moreover, the
proposed model retains optimality. For an optimal solution $\mathbf{t}^{*}$
that has a dimension less than $M$, corresponding to less than $M-1$
hops, the dynamic space-time graph model provides an $M$-dimension
equivalent solution $\mathbf{t}^{\prime}$ by assigning some virtual
edges. For example, if the optimal route is $1\to4$, with time boundary
$\mathbf{t}^{*}=(t_{1},t_{4})$, then under $M=4$, the dynamic space-time
graph model can provide a fixed-dimension solution $1\to1\to1\to4$,
with $\mathbf{t}^{\prime}=(t_{1},t_{1},t_{1},t_{4})$, as shown in
Fig.~\ref{fig:case_example_route} (c). Since the virtual edge has
zero weight, {\em i.e.}, $w_{1,1}^{k}=0$, the fixed-dimension dynamic
space-time graph model yields the same minimum cost as the classical
model with the optimal graph dimension.

In addition, setting the graph dimension equal to the number of nodes
$M$ is sufficient, and the minimum required, to guarantee optimality.
This is because the graph has $M$ nodes, and thus, any path with
more than $M-1$ edges must form a loop; for a path with a loop that
passes the $m$th node twice, it is more efficient to simply stay
on the $m$th node for a longer time. For example, for a route $1\to2\to1\to4$,
with time boundary $(t_{1},t_{2},t_{3},t_{4})$, it is more efficient
to reduce the route to $1\to4$ with the time boundary $(t_{1},t_{4})$,
as shown in Fig.~\ref{fig:case_example_route}. As a result, setting
the graph dimension to $M$ is sufficient to capture all optimal paths.
In addition, the dimension should not be smaller than $M$ in order
to accommodate paths that consist of exactly $M-1$ edges.

\subsection{Problem Formulation}

Based on the dynamic space-time graph model, this paper focuses on
network-level optimization, aiming to optimize the time boundary $\mathbf{t}$
of the transportation of the data package and the route $\mathbf{o}\triangleq\{o(1),\cdots,o(M)\}$
from the source $o(1)=1$ to the destination $o(M)=M$ for $T$ time
ahead for the transmission that minimizes the worst-case interference
power leakage $\vartheta$ during full data transportation
\begin{align}
\mathscr{P}2:\quad\underset{\mathbf{o},\mathbf{t},\vartheta}{\text{minimize}} & \quad\vartheta\label{eq:obj_p1}\\
\text{subject to} & \quad w_{o(k),o(k+1)}^{k}\left(t_{k},t_{k+1}\right)\le\vartheta,\forall k\label{eq:c_intf_p1}\\
 & \quad o\left(1\right)=1,o(M)=M,o\left(k\right)\in\mathcal{M}\label{eq:c_relay_p1}\\
 & \quad0=t_{1}\le\cdots\le t_{M}=T\label{eq:c_ht_p1}
\end{align}
where constraint (\ref{eq:c_intf_p1}) ensures that the interference
power leakage of each relay is less than $\vartheta$ (also the epigraph
form of (\ref{eq:def_net_itf})), constraint (\ref{eq:c_relay_p1})
ensures that the data is transmitted from the source node $o(1)=1$
to the destination node $o(M)=M$ through the aerial communication
network $o(k)\in\mathcal{M}$, and (\ref{eq:c_ht_p1}) is the time
causality constraint, ensuring that data is fully transferred from
one node to another before it is forwarded to a third node.

\section{Graph-Based Cross-Layer Optimization for Single Commodity Transportation\label{sec:alg_single}}

Given the time boundary $\mathbf{t}$ and weights $w_{m,n}^{k}(t_{k},t_{k+1})$,
problem $\mathscr{P}2$ becomes a conventional minimax path problem,
where the optimal route $\mathbf{o}$ can be found by a bottleneck
path planning algorithm \cite{RanSet:C09,DuaLyuWuXie:C18}. Thus,
the key challenges are to find the weights $w_{m,n}^{k}(t_{k},t_{k+1})$
by efficiently solving the inner problem $\mathscr{P}1$ and find
the optimal time boundary $\mathbf{t}$.

First, solving $\mathscr{P}1$ involves optimizing the power allocation
policy $p_{m,n}(t)$ under the uncertainty of the future instantaneous
channel quality $h_{m,n}(t)$ in a horizon of $t\in(0,T)$. Mathematically,
this requires solving a problem with expectation without a closed-form
expression as in (\ref{eq:c_thp}). To address this, by analyzing
the optimality condition of $\mathscr{P}1$, we derive a closed-form
lower bound and approximations to explicitly evaluate the constraint
(\ref{eq:c_thp}), which facilitate the design of an efficient algorithm.

Second, solving for the time boundary $\mathbf{t}$ requires optimization
over a non-convex feasible set of $\{\mathbf{t},\vartheta\}$. We
discover a monotonicity of the time boundary over $\vartheta$, which
enables an efficient bisection search for the global optimum $\mathbf{t}$
over each route $\mathbf{o}$.

\subsection{Inner Problem: Power Allocation Policy\label{subsec:Inner-Problem}}

It is found that if the distribution of the channel $h_{m,n}(t)$
is available, the optimal power allocation policy $p_{m,n}(t)$ can
be found using Lagrangian methods.

By equivalently reformulating constraint (\ref{eq:c_itf}) of problem
$\mathscr{P}1$ as
\begin{equation}
p_{m,n}\left(t\right)\max_{j\in\mathcal{N}}\left\{ h_{m,j}\left(t\right)\right\} \le\vartheta,\forall t\in\left[t_{k},t_{k+1}\right)\label{eq:itf_c_rv1}
\end{equation}
$\mathscr{P}1$ can be equivalently reformulated as
\[
\underset{\left\{ p_{m,n}\left(t\right)\right\} ,\vartheta_{m,n}}{\text{minimize}}\vartheta_{m,n},\ \text{s.t. (\ref{eq:c_thp}) and (\ref{eq:itf_c_rv1}).}
\]
The reformulated problem is convex, since the objective function and
constraint (\ref{eq:itf_c_rv1}) are linear, and constraint (\ref{eq:c_thp})
is convex. Therefore, the \ac{kkt} conditions of the reformulated
problem are the sufficient optimality conditions of problem $\mathscr{P}1$
\cite{boyd2004convex}.

Let $\boldsymbol{\Lambda}=[\mathbf{v},\mu]$ be the Lagrangian set,
then the Lagrangian function of problem $\mathscr{P}1$ can be expressed
as 
\begin{align*}
 & L\left(p_{m,n}\left(t\right),\vartheta_{m,n},\boldsymbol{\Lambda}\right)=\\
 & \vartheta_{m,n}+\int_{t_{k}}^{t_{k+1}}v\left(t\right)\left(p_{m,n}\left(t\right)\max_{j\in\mathcal{N}}\left\{ h_{m,j}\left(t\right)\right\} -\vartheta_{m,n}\right)dt\\
 & +\mu\left(S-\int_{t_{k}}^{t_{k+1}}\mathbb{E}\left[B\log_{2}\left(1+\frac{p_{m,n}\left(t\right)h_{m,n}\left(t\right)}{\delta^{2}}\right)\right]dt\right).
\end{align*}
Let $(p_{m,n}^{*}(t),\vartheta_{m,n}^{*})$ be the optimal solution
to $\mathscr{P}1$ and $\boldsymbol{\Lambda}^{*}$ be the optimal
Lagrange multiplier for its dual problem. From the \ac{kkt} optimality
conditions, we derive that $(p_{m,n}^{*}(t),\vartheta_{m,n}^{*},\boldsymbol{\Lambda}^{*})$
should satisfy primal feasibility (\ref{eq:c_thp}) and (\ref{eq:itf_c_rv1}),
dual feasibility 
\begin{equation}
\mu\ge0,v\left(t\right)\ge0,\forall t\label{eq:KKT_c1}
\end{equation}
complementary slackness 
\begin{align}
 & v\left(t\right)\left(p_{m,n}\left(t\right)\max_{j\in\mathcal{N}}\left\{ h_{m,j}\left(t\right)\right\} -\vartheta_{m,n}\right)=0,\forall t\label{eq:KKT_c2}\\
 & \mu\left(S-\int_{t_{k}}^{t_{k+1}}\mathbb{E}\left[B\log_{2}\left(1+\frac{p_{m,n}\left(t\right)h_{m,n}\left(t\right)}{\delta^{2}}\right)\right]dt\right)=0\label{eq:KKT_c3}
\end{align}
and stationarity
\begin{align}
 & \frac{\partial L}{\partial\vartheta_{m,n}}=1-\int_{t_{k}}^{t_{k+1}}v\left(t\right)dt=0\label{eq:KKT_c4}\\
 & \frac{\partial L}{\partial p_{m,n}\left(t\right)}=v\left(t\right)\max_{j\in\mathcal{N}}\left\{ h_{m,j}\left(t\right)\right\} \nonumber \\
 & \quad\quad\quad\quad-\mu\mathbb{E}\left[\frac{B}{\ln2}\frac{h_{m,n}\left(t\right)}{\delta^{2}+p_{m,n}\left(t\right)h_{m,n}\left(t\right)}\right]=0,\forall t\label{eq:KKT_c5}
\end{align}

By analyzing the \ac{kkt} conditions (\ref{eq:c_thp}) and (\ref{eq:itf_c_rv1})\textendash (\ref{eq:KKT_c5}),
the optimal power allocation $p_{m,n}^{*}(t)$ can be found as follows.
\begin{prop}
\label{prop:power_allocation_policy}(Optimal solution to $\mathscr{P}1$)
Given the time boundaries, the optimal value of $\mathscr{P}1$ is
$w_{m,n}^{k}=\vartheta_{m,n}^{*}$, which is the solution to
\begin{equation}
\Upsilon_{m,n}\left(\vartheta_{m,n};t_{k},t_{k+1}\right)=S\label{eq:eqv_obj_p1}
\end{equation}
and $\Upsilon_{m,n}(\vartheta_{m,n},t_{k},t_{k+1})$ is defined as
\begin{multline}
\Upsilon_{m,n}\left(\vartheta_{m,n},t_{k},t_{k+1}\right)\triangleq\\
\int_{t_{k}}^{t_{k+1}}\mathbb{E}\left[B\log_{2}\left(1+\frac{\vartheta_{m,n}h_{m,n}\left(t\right)}{\max_{j\in\mathcal{N}}\left\{ h_{m,j}\left(t\right)\right\} \delta^{2}}\right)\right]dt.\label{eq:sum_rate_def}
\end{multline}
In addition, the optimal power allocation policy $p_{m,n}^{*}\left(t\right)$
is given by 
\begin{equation}
p_{m,n}^{*}\left(t\right)=\frac{\vartheta_{m,n}^{*}}{\max_{j\in\mathcal{N}}\left\{ h_{m,j}\left(t\right)\right\} }\label{eq:opt_p_0}
\end{equation}
\end{prop}
\begin{proof}
See Appendix~\ref{sec:proof_power_allocation_policy}.
\end{proof}
It is observed that $\Upsilon_{m,n}(\vartheta_{m,n};t_{k},t_{k+1})$
in (\ref{eq:sum_rate_def}) is strictly increasing over $\vartheta_{m,n}$.
Therefore, the optimal parameter $\vartheta_{m,n}$ to satisfy (\ref{eq:eqv_obj_p1})
can be found using bisection search and the optimal solution is unique.

\subsection{Deterministic Bound}

The challenge of solving (\ref{eq:eqv_obj_p1}) is the efficient evaluation
of the expectation, which does not have a closed-form expression.
In the following proposition, we derive a lower bound for the expected
channel capacity as the integrand in (\ref{eq:sum_rate_def}).
\begin{prop}
\label{prop:lb_c}(A deterministic capacity lower bound) The expected
capacity in (\ref{eq:sum_rate_def}) 
\[
\mathbb{E}\left[B\log_{2}\left(1+\frac{\vartheta_{m,n}h_{m,n}\left(t\right)}{\max_{j\in\mathcal{N}}\left\{ h_{m,j}\left(t\right)\right\} \delta^{2}}\right)\right]
\]
is lower bounded by
\begin{align}
\underline{c}_{m,n}\left(t\right) & \triangleq\log_{2}\left(1+\frac{\vartheta g_{m,n}\left(t\right)}{\left(\max_{j\in\mathcal{N}}\left\{ g_{m,j}\left(t\right)\right\} +\alpha\omega_{m}\left(t\right)\right)\delta^{2}}\right)\nonumber \\
 & \quad-\epsilon_{m,n}\left(t\right)\label{eq:c_l_b}
\end{align}
where $\alpha=1$, $\omega_{m}(t)=\sqrt{\sum_{j\in\mathcal{N}}g_{m,j}(t)^{2}/\kappa_{m,j}(t)}$,
and $\epsilon_{m,n}(t)=\log_{2}(e)/\kappa_{m,n}(t)-\log_{2}(1+(2\kappa_{m,n}(t))^{-1})$.
\end{prop}
\begin{proof}
See Appendix~\ref{sec:proof:lb_c}.
\end{proof}
The gap between the expected capacity $\mathbb{E}[c_{m,n}(t)]$ and
its lower bound $\underline{c}_{m,n}\left(t\right)$ arises from the
positive $\omega_{m}(t)$ and $\epsilon_{m,n}(t)$, and decreases
with $\kappa$ increasing, tending to $0$ when the parameter $\kappa$
in the Gamma distribution of small-scale fading in (\ref{eq:channel_model})
goes to infinity in the \ac{los} case, as shown in Fig.~\ref{fig:lower_bound}.

\begin{figure}
\begin{centering}
\includegraphics[width=0.9\columnwidth]{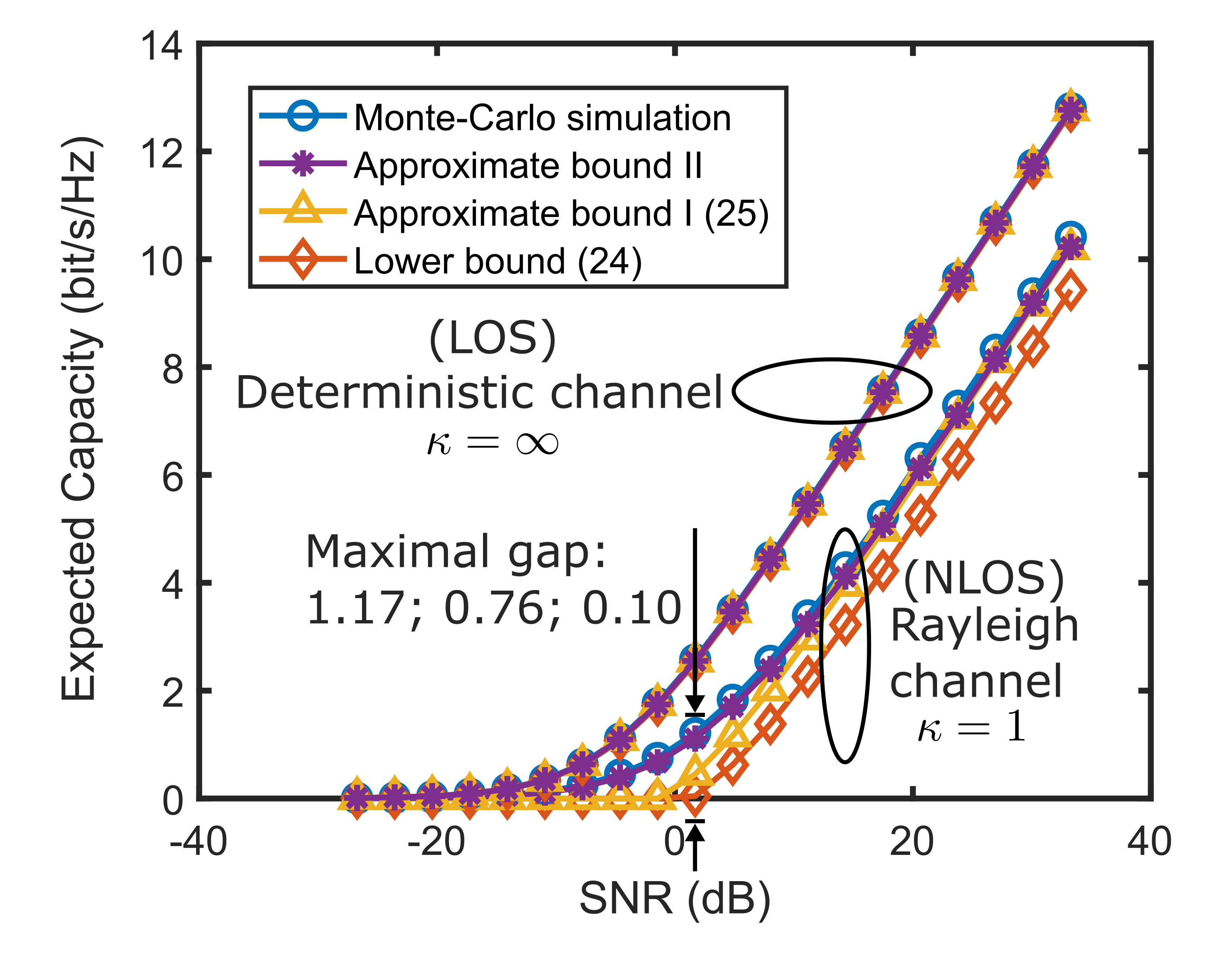}
\par\end{centering}
\caption{\label{fig:lower_bound}Comparison of different capacity lower bounds.
The group with $\kappa=\infty$ corresponds to a deterministic channel,
while the group with $\kappa=1$ corresponds to a Rayleigh fading
channel. The maximum gaps occur at $\kappa=1$, where the differences
between the Monte Carlo result and the lower bound (\ref{eq:c_l_b}),
approximate bound I, and approximate bound II are 1.17, 0.076, and
0.10\,dB, respectively.}
\end{figure}

It is observed that the lower bound (\ref{eq:c_l_b}) is tight in
the \ac{los} case when $\kappa\to\infty$. In the \ac{nlos} case,
there is a 0.43 dB gap in the high \ac{snr} regime, and gap becomes
larger in the low \ac{snr} regime. To find a tighter approximation,
we derive two expressions as follows.

First, we numerically found that setting $\alpha=1/2$ in (\ref{eq:c_l_b})
yields a tighter bound, referred as Approximate bound I, 
\begin{equation}
\log_{2}\left(1+\frac{\vartheta g_{m,n}\left(t\right)}{\left(\max_{j\in\mathcal{N}}\left\{ g_{m,j}\left(t\right)\right\} +\frac{1}{2}\omega_{m}\left(t\right)\right)\delta^{2}}\right)-\epsilon_{m,n}\left(t\right)\label{eq:aprox_bound_1}
\end{equation}
which appears to still be a lower bound over all $\kappa$ we tested.
As shown in Fig.~\ref{fig:lower_bound}, this approximate bound aligns
closely with the Monte Carlo empirical value in high \ac{snr} scenarios.

Second, to obtain a tighter bound in the low SNR regime, we employ
Jensen's inequality sharpening techniques \cite{LiaBer:J19} and obtain
Approximate Bound II as $f(\gamma;\kappa)$, where

\[
\gamma=\frac{\vartheta g_{m,n}\left(t\right)}{\left(\max_{j\in\mathcal{N}}\left\{ g_{m,j}\left(t\right)\right\} +\frac{1}{2}\omega_{m}\left(t\right)\right)\delta^{2}}
\]
and 
\begin{equation}
f\left(\gamma;\kappa\right)\triangleq\mathbb{E}\left[\log_{2}\left(1+\gamma\xi\right)\right]\label{eq:aprox_bound_2}
\end{equation}
in which $\xi\sim\text{Gamma}(\kappa,1/\kappa)$. Note that the function
$f(\gamma;\kappa)$, which has been substantially simplified from
(\ref{eq:sum_rate_def}), can be computed offline and stored in a
table. Hence, Approximate Bound II can be computed efficiently in
closed-form with table lookup.

Simulation results in Fig.~\ref{fig:lower_bound} demonstrate that
Approximate Bound II consistently aligns with Monte Carlo simulations,
from low to high \ac{snr} conditions.\footnote{Approximate Bound I also aligns with Monte Carlo simulations when
the targeted transmission channel is LOS (regardless of the interfering
channels), because in this case, Approximate Bound I is exactly equal
to Approximate Bound II, as shown in Appendix~\ref{sec:proof:lb_c}.}

\subsection{Outer Problem: Time Boundary Optimization}

The time boundary optimization finds the optimal time boundary $\mathbf{t}=(t_{1},t_{2},\cdots,t_{M})$
in $\mathscr{P}2$ with a given route $\mathbf{o}$ and the weights
$w_{m,n}^{k}$ given in Proposition~\ref{prop:power_allocation_policy}.
It can be verified that the feasible set of $\mathbf{t}$ defined
by the constraint (\ref{eq:c_intf_p1}) is non-convex. Despite the
non-convexity, we will show that the optimal $t_{k}$'s can be uniquely
determined by $t_{M}$, for $k<M$, and the objective value $\vartheta$
in $\mathscr{P}2$ is monotonic in $t_{M}$. This result can lead
to a globally optimal solution to $\mathbf{t}$.
\begin{prop}
\label{prop:opt_c_p2_over_o}(Optimality condition to $\mathscr{P}2$)
Given the route variable $\mathbf{o}$, $\{\mathbf{t}^{*},\vartheta^{*}\}$
is the optimal solution to problem $\mathscr{P}2$ if and only if
$\vartheta^{*}=w_{o(k),o(k+1)}^{k}(t_{k}^{*},t_{k+1}^{*})$ for all
$k\in\{1,\cdots,M-1\}$.
\end{prop}
\begin{proof}
See Appendix~\ref{sec:proof_opt_c_p2_over_o}.
\end{proof}
Proposition~\ref{prop:opt_c_p2_over_o} states that, as a sufficient
and necessary optimality condition, all the links should achieve the
same interference leakage level $w_{m,n}^{k}=\vartheta^{*}$ along
the route $\mathbf{o}$. The intuition is that if the $k$th link
along the route has a smaller interference leakage power than that
of the $(k+1)$th link, then a portion of resources of the $k$th
link can be reallocated to the $(k+1)$th link by adjusting the variable
$t_{k+1}$, and hence, the solution is not optimal.

Using Propositions~\ref{prop:power_allocation_policy} and \ref{prop:opt_c_p2_over_o},
we can express the optimal time boundary $\bm{t}^{*}$ as a function
of the power leakage $\vartheta$ as follows. First, we note that
$t_{1}=0$ must be the optimal solution according to constraint (\ref{eq:c_ht_p1})
in $\mathscr{P}2$. Second, given $t_{k}$ for $k\leq M-1$, the optimal
$t_{k+1}$ can be obtained as the solution to $\Upsilon_{m,n}(\vartheta;t_{k},t_{k+1})=S$
in (\ref{eq:eqv_obj_p1}). Such an expression implies that all the
links are to achieve the same power leakage $\vartheta$. Then, using
induction, the last time boundary $t_{M}(\vartheta)$ can be obtained.
Finally, by the definition of $\mathscr{P}2$, if $t_{M}(\vartheta)=T$,
then $\vartheta$ must be the optimal solution; and if $t_{M}(\vartheta)\neq T$,
$\vartheta$ is not the solution, because the constraint (\ref{eq:c_ht_p1})
in $\mathscr{P}2$ is violated.

In case $t_{M}(\vartheta)\neq T$, we find that $t_{M}(\vartheta)$
is monotonic.
\begin{prop}
\label{prop:c_increasing_over_vartheta}(Monotonicity of $t_{M}(\vartheta)$)
The function $t_{M}(\vartheta)$ is strictly decreasing over $\vartheta$.
\end{prop}
\begin{proof}
See Appendix~\ref{sec:proof_c_increasing_over_vartheta}.
\end{proof}
Proposition~\ref{prop:c_increasing_over_vartheta} states that increasing
the tolerance of the interference leakage $\vartheta$, less time
resource is needed and $t_{M}(\vartheta)$ is reduced.

The monotonicity property implies a bisection search strategy to find
the optimal $\vartheta^{*}$ such that $t_{M}(\vartheta^{*})=T$.
This leads to the basic structure of Algorithm~\ref{alg:ht_opt_bi_alg}.

\begin{algorithm}
\# Input: Route $\mathbf{o}$ and time boundary $\hat{\mathbf{t}}=\{\hat{t}_{k}\}_{k\in\{1,\cdots,M\}}$.
\begin{enumerate}
\item Set $\vartheta\leftarrow(\vartheta_{\max}+\vartheta_{\min})/2$.
\item Starting from $t_{1}=0$, update $t_{k+1}\leftarrow t_{k+1}^{\prime}+\alpha_{k}(\hat{t}_{k+1}-\hat{t}_{k})$,
where $t_{k+1}^{\prime}$ is computed based on $t_{k}$ by solving
$\Upsilon_{o(k),o(k+1)}\left(\vartheta;t_{k},t_{k+1}\right)=S$ in
(\ref{eq:eqv_obj_p1}) and $\alpha$ is defined in (\ref{eq:def_alpha}).
\item If $t_{M}(\vartheta)\le T$, $\vartheta_{\max}\leftarrow\vartheta$;
otherwise, $\vartheta_{\min}\leftarrow\vartheta$.
\item Repeat from step 1) until $\left|\vartheta_{\max}-\vartheta_{\min}\right|\to0$.
\end{enumerate}
\# Output: $\bm{t}$ and $p_{m,n}(t)$ from (\ref{eq:opt_p_0}).

\caption{\label{alg:ht_opt_bi_alg}Time and power allocation algorithm.}
\end{algorithm}

\subsection{Implementation with Backtracking\label{subsec:Graph-based-hybrid-optimization}}

In order to improve the convergence, a backtracking scheme is used
in Algorithm~\ref{alg:ht_opt_bi_alg}. Specially in step 2) in Algorithm~\ref{alg:ht_opt_bi_alg},
update $t_{k+1}\leftarrow t_{k+1}^{\prime}+\alpha_{k}(\hat{t}_{k+1}-\hat{t}_{k})$,
where $(\hat{t}_{k+1}-\hat{t}_{k})$ is the past transmission duration
for $k$th hop, serving as an external parameter that remains unchanged
in Algorithm~\ref{alg:ht_opt_bi_alg}, and $\alpha_{k}$ is the backtracking
parameter, defined as 
\begin{equation}
\alpha_{k}=\alpha\mathbb{I}\left\{ o\left(k\right)=o\left(k+1\right)\right\} ,\,\alpha\in\left(0,1\right).\label{eq:def_alpha}
\end{equation}
Backtracking only performs when the virtue edge is selected.

The reason that a backtracking strategy is employed to update $t_{k}$
is to prevent being trapped at a local optimum too early and to improve
the convergence. It is known that an alternating algorithm is prone
to being trapped at a stationary point if not appropriately initialized,
and a soft update can relieve this phenomenon. In our case, a virtual
edge implies $t_{k}=t_{k+1}$ and the $k$th hop should be effectively
removed. However, removing a layer in the graph will permanently prevents
adding back this layer in future iterations. Thus, the backtracking
update prevents the collapse of the graph while still allowing $t_{k+1}\to t_{k}$
for an virtual edge.

Since the construction of $\mathbf{t}$ from $\vartheta$ is unique
due to Proposition~\ref{prop:power_allocation_policy} and the optimality
condition given in Proposition~\ref{prop:opt_c_p2_over_o} is both
sufficient and necessary, it naturally leads to the following optimality
result.
\begin{prop}
\label{prop:opt_alg_1}(Optimality) When the backtracking finishes,
{\em i.e.}, $\alpha_{k}(\hat{t}_{k+1}-\hat{t}_{k})=0,\forall k$,
Algorithm~\ref{alg:ht_opt_bi_alg} finds the optimal solution \textup{$\mathbf{t}^{*}$}
to $\mathscr{P}2$ given each $\mathbf{o}$.
\end{prop}
The overall structure of solving $\mathscr{P}2$ is summarized in
the following looping steps:
\begin{enumerate}
\item[i)] update $\mathbf{o}$ from bottleneck path planning
\item[ii)] compute $\mathbf{t}$ from Algorithm~\ref{alg:ht_opt_bi_alg}.
\end{enumerate}
A detailed implementation is given in Algorithm~\ref{alg:hybrid_opt_alg}.
We show in the following that the Algorithm~\ref{alg:hybrid_opt_alg}
converges and its computational complexity is proportional to the
duration and the cube of the number of aerial nodes.

\subsubsection{Convergence analysis\label{subsec:Convergence-analysis}}

The following analysis proves that Algorithm~\ref{alg:hybrid_opt_alg}
converges, and that all virtual edges are eliminated at convergence;
specifically, $t_{k+1}-t_{k}=0$ whenever $o(k+1)=o(k)$.

Algorithm~\ref{alg:hybrid_opt_alg} modifies the objective value
$\vartheta$ in two places: updating the route $\mathbf{o}$ in step
2) and updating the time boundary in step 3). It is proven below that
the objective value $\vartheta$ decreases through the iteration.
Furthermore, $\vartheta$ is bounded below by 0, then the Algorithm~\ref{alg:hybrid_opt_alg}
must converge.

The objective value $\vartheta$ decreases during the route update
$\mathbf{o}$ in step 2) because, for a given time boundary $\mathbf{t}^{(i-1)}$,
the previous transmission route $\mathbf{o}^{(i-1)}$ is a feasible
solution to problem $\mathscr{P}2$ and the bottleneck path planning
algorithm can find the optimal solution to $\mathscr{P}2$, which
ensures find a lower $\vartheta$. Similarly, the objective value
$\vartheta$ decreases during the time boundary update $\mathbf{t}$
in step 3). For the current route $\mathbf{o}^{(i)}$, the objective
value for the previous time boundary $\mathbf{t}^{(i-1)}$ is greater
than or equal to that of a modified time boundary $\mathbf{t}^{\prime}$,
where $t_{k+1}^{\prime}=t_{k}^{(i-1)}+\alpha(t_{k+1}^{(i-1)}-t_{k}^{(i-1)})$
if $o^{(i)}(k)=o^{(i)}(k+1)$ and $t_{k+1}^{\prime}=t_{k+1}^{(i-1)}$
otherwise. This holds because when $o^{(i)}(k)=o^{(i)}(k+1)$, we
have
\[
w_{o^{(i)}(k),o^{(i)}(k+1)}(t_{k}^{\prime},t_{k+1}^{\prime})=0
\]
 due to the property of virtual edge, and 
\begin{multline*}
w_{o^{(i)}(k+1),o^{(i)}(k+2)}(t_{k+1}^{\prime},t_{k+2}^{\prime})\\
<w_{o^{(i)}(k+1),o^{(i)}(k+2)}(t_{k+1}^{(i-1)},t_{k+2}^{(i-1)})
\end{multline*}
due to the decreasing monotonicity of $w_{m,n}(t_{k},t_{k+1})$ over
$t_{k}$ as stated in Lemma~\ref{lem:monotonicity} in Appendix~\ref{sec:proof_opt_c_p2_over_o}.
Since $\mathbf{t}^{\prime}$ is a feasible solution to problem $\mathscr{P}2$,
and the Algorithm~\ref{alg:ht_opt_bi_alg} identifies the optimal
solution to $\mathscr{P}2$, similarly to Proposition~\ref{prop:opt_alg_1},
updating the time boundary in step 3) guarantees a lower $\vartheta$.

All virtual edges are eliminated upon convergence because the time
interval of the selected virtual edge after the $i$th iteration is
$\alpha^{i}(t_{k}^{(i-1)}-t_{k-1}^{(i-1)})<\left|\mathbf{t}^{(i)}-\mathbf{t}^{(i-1)}\right|$,
which approaches zero as $\left|\mathbf{t}^{(i)}-\mathbf{t}^{(i-1)}\right|\to0$
when convergence.

\subsubsection{Complexity analysis\label{subsec:complexity_analysis_single}}

The overall computational complexity for Algorithm~\ref{alg:hybrid_opt_alg}
is $\mathcal{O}(M^{2}(M+T)\omega)$, where $\omega$ is the iteration
number of hybrid optimization algorithm, including $\mathcal{O}(M^{2}T)$
for graph construction in step 1), $\mathcal{O}(M^{3})$ for bottleneck
path selection in step 2), and $\mathcal{O}(T)$ for time boundaries
update in step 3). In detail, graph construction includes $M^{3}$
weight calculations, each requiring $\mathcal{O}(t_{k+1}-t_{k})$,
leading to a total complexity of $\mathcal{O}(M^{2}T)$. The route
update in graph $\mathscr{G}\left(\mathbf{t}\right)$ with $M$ incoming
edges per node, $M$ nodes per layer and $M$ layers, resulting in
a complexity of $\mathcal{O}(M^{3})$. The time boundaries update
includes $\mathcal{O}(1)$ of searching optimal $\vartheta^{*}$ by
$M-1$ times of calculation for $t_{k}(\vartheta)$ with complexity
$\mathcal{O}(t_{k+1}-t_{k})$, therefore time complexity is $\mathcal{O}(T)$.

\begin{algorithm}
\# Initialization: Set $\mathbf{t}^{(0)}\leftarrow\{kT/(M-1)\}_{k\in\{0,\cdots,M-1\}}$
and $i\leftarrow1$.
\begin{enumerate}
\item Solve (\ref{eq:eqv_obj_p1}) for $\vartheta_{m,n}^{*}$ and obtain
$w_{m,n}^{k}=\vartheta_{m,n}^{*}$, where approximations (\ref{eq:c_l_b})\textendash (\ref{eq:aprox_bound_2})
can be used.
\item Use a bottleneck path planning algorithm to obtain $\mathbf{o}^{(i)}$
by solving $\mathscr{P}2$.
\item Use Algorithm~\ref{alg:ht_opt_bi_alg} to obtain $\mathbf{\ensuremath{t}}^{(i)}$
where approximations (\ref{eq:c_l_b})\textendash (\ref{eq:aprox_bound_2})
can be used.
\item Repeat from step 1) until $\left|\mathbf{t}^{(i)}-\mathbf{t}^{(i-1)}\right|\to0$.
\end{enumerate}
\# Output: $\mathbf{t}^{(i)}$, $\mathbf{o}^{(i)}$, and $p_{m,n}(t)$
from (\ref{eq:opt_p_0}).

\caption{\label{alg:hybrid_opt_alg}Efficient graph-based single commodity
algorithm.}
\end{algorithm}

\section{Multi-commodity Transportation\label{sec:multi_task}}

This section extends the single commodity transportation in Section
\ref{sec:alg_single} to multi-commodity transportation, where a commodity
refers to a data package of size $S$ transported from a source node
to a destination node within a deadline of $T$ seconds. For simplicity,
we assume all the commodities have the same size $S$ and deadline
$T$. Consider that orthogonal time-frequency resources are dynamically
allocated to transmitting different commodities so that there is no
interference among the nodes in $\mathcal{M}$ in the aerial network,
but there is still interference to the neighbor nodes $\mathcal{N}$.
Therefore, the core problem is to orthogonally allocate the time-frequency
resources in a predictive way for a horizon of $T$ seconds to exploit
the dynamic of the network topology while controlling the possible
interference to nodes in $\mathcal{N}$.

\subsection{Multi-Commodity Transportation Problem Formulation\label{subsec:multi_commodity_model}}

Consider there are $Z$ data packages, each with a size $S$, that
needs to be delivered from source node $s_{z}$ to the destination
node $d_{z}$, where $z\in\mathcal{Z}\triangleq\{1,2,\cdots,Z\}$.
Consider multiple tasks share orthogonal time-frequency resource and
each flow $z\in\mathcal{Z}$ occupies $l_{z}\left(t\right)$ normalized
resource at time $t$. Therefore, the resource allocation strategy
$\boldsymbol{l}(t)\triangleq\{l_{z}(t)\}_{z\in\mathcal{Z}}$ at each
time $t$ should satisfy the orthogonality constraint $\boldsymbol{l}(t)\in\mathcal{L}$,
where
\[
\mathcal{L}\triangleq\left\{ \left\{ l_{z}\left(t\right)\right\} _{z\in\mathcal{Z}}:l_{z}\left(t\right)\in\left[0,1\right],\sum_{z\in\mathcal{Z}}l_{z}\left(t\right)\in\left[0,1\right]\right\} .
\]

Our goal is to transport all data packages while minimizing the maximum
interference power leakage during the process. This is achieved by
controlling the transmission route $\mathbf{O}\triangleq\{o(k,z)\}_{k\in\mathcal{M},z\in\mathcal{Z}}$,
time boundaries $\mathbf{T}\triangleq\{t_{k,z}\}_{k\in\mathcal{M},z\in\mathcal{Z}}$,
time-frequency resource allocation $\mathbf{L}\triangleq\{\boldsymbol{l}(t)\}_{t\in[0,T]}$,
and power allocation strategy $\mathbf{P}\triangleq\{p_{o(k,z),o(k+1,z)}(t)\}_{k\in\mathcal{M},z\in\mathcal{Z},t\in[0,T]}$
for all the commodities $z$ along each hop $k$. Then, the problem
is formulated as
\begin{align}
\mathscr{P}3:\;\underset{\mathbf{O},\mathbf{T},\mathbf{L},\vartheta}{\text{minimize}} & \quad\vartheta\label{eq:multi_task_obj_p1}\\
\text{subject to} & \quad w_{o(k,z),o(k+1,z)}^{k}\left(t_{k,z},t_{k+1,z},\mathbf{l}_{z}\right)\le\vartheta,\forall k,z\label{eq:max_itf_c}\\
 & \quad o\left(1,z\right)=s_{z},o\left(M,z\right)=d_{z},\forall z\label{eq:multi_o_c_1}\\
 & \quad o\left(k,z\right)\in\mathcal{M},\forall k,z\label{eq:multi_o_c_2}\\
 & \quad0\le t_{1,z}\le\cdots\le t_{M,z}\le T,\forall z\label{eq:multi_t_c}\\
 & \quad\boldsymbol{l}(t)\in\mathcal{L},\forall t\label{eq:multi_l_c}
\end{align}
where the objective (\ref{eq:multi_task_obj_p1}) is the maximum interference
power leakage and constraint (\ref{eq:max_itf_c}) is to ensure that
the interference to any neighbor node during the transmission of all
$Z$ data is less than $\vartheta$. Constraints (\ref{eq:multi_o_c_1})\textendash (\ref{eq:multi_l_c})
are the relay, time causality, and time-frequency resource constraints.

The weight $w_{m,n,z}^{k}\left(t_{k,z},t_{k+1,z},\mathbf{l}_{z}\right)$
is the maximum interference to the neighbor network during the $k$th
hop under the allocated time-frequency resource $\mathbf{l}_{z}\triangleq\{l_{z}(t)\}_{t\in[0,T]}$
defined according to $\mathscr{P}1$, where the total bandwidth $B$
in (\ref{eq:c_thp}) is replaced as a sub-band bandwidth $l_{z}\left(t\right)\cdot B$.

To address the challenge of coupled variables, we first decompose
problem $\mathscr{P}3$ into several parallel single-commodity subproblems
with fixed time-frequency allocations in Section \ref{subsec:multi_to_single}.
Then in in Section \ref{subsec:Resource-Allocation-Optimization},
we introduce a simplex-based bisection search algorithm to determine
the optimal time-frequency allocation , even though the resulting
subproblem remains non-convex.

\subsection{Problem Decomposition and Relationship to the Single-Commodity Case\label{subsec:multi_to_single}}

It is observed from problem $\mathscr{P}3$ that for multiple tasks,
the variables are coupled over $z$ only by constraint (\ref{eq:max_itf_c}).
Therefore, given the time-frequency resource allocation variable $\mathbf{L}$,
problem $\mathscr{P}3$ is decomposed into a number of parallel single
commodity subproblems identical to $\mathscr{P}2$ that have been
solved in Section \ref{sec:alg_single}.

Specifically, denote $\vartheta_{z}$ as the maximum interference
to the neighbor network during the transportation of the specific
data package $z$, then, $\vartheta_{z}$ is the optimal solution
to Problem $\mathscr{P}2$ with allocated resource $l_{z}\left(t\right)\cdot B$
and the solution is given by Algorithm~\ref{alg:hybrid_opt_alg}.
As a result, given time-frequency resource allocation $\mathbf{L}$,
the optimal interference power leakage is $\vartheta=\max_{z\in\mathcal{Z}}\{\vartheta_{z}\}$.

Similar to the optimal power policy for problem $\mathscr{P}1$ discussed
in Proposition~\ref{prop:power_allocation_policy}, the optimal power
allocation for problem $\mathscr{P}3$ is given by 
\begin{equation}
p_{o\left(k,z\right),o\left(k+1,z\right)}\left(t\right)=\vartheta/\max_{j\in\mathcal{N}}\left\{ h_{o\left(k,z\right),j}\left(t\right)\right\} .\label{eq:opt_p_multi}
\end{equation}

The insight is that, for any given interference power leakage $\vartheta$,
the maximum power allowed under the interference constraint is utilized
to maximize the throughput.

\subsection{Resource Allocation Optimization\label{subsec:Resource-Allocation-Optimization}}

Next, we investigate an efficient algorithm for finding the resource
allocation $\mathbf{L}$, which leads to a non-convex problem.

Given route and time boundary variables $\{\mathbf{O},\mathbf{T}\}$,
problem $\mathscr{P}3$ over variables $\{\mathbf{L},\vartheta\}$
is simplified to:
\begin{align}
\underset{\mathbf{L},\vartheta}{\text{minimize}} & \ \vartheta\label{eq:opt_l_v2}\\
\text{subject to} & \ \int_{t_{k,z}}^{t_{k+1,z}}\mathbb{E}\left[\log\left(1+\frac{\vartheta h_{o\left(k,z\right),o\left(k+1,z\right)}\left(t\right)}{\max_{j\in\mathcal{N}}\left\{ h_{o\left(k,z\right),j}\left(t\right)\right\} \delta^{2}}\right)\right]\nonumber \\
 & \ \quad\quad\ \times l_{z}\left(t\right)Bdt\ge S,\forall z\in\mathcal{Z},\forall k\in\mathcal{M}\label{eq:multi_thp_c_v2}\\
 & \ \boldsymbol{l}(t)\in\mathcal{L},\forall t.\label{eq:multi_l_c_v2_2}
\end{align}

This problem is non-convex because of the non-convexity of the throughput
constraint in (\ref{eq:multi_thp_c_v2}). However, we observe a monotonicity
property on the region of the feasible time-frequency allocation variable,
which enables a bisection-search method to find the optimal solution
$\mathbf{L}$ to $\mathscr{P}3$.

Denote the feasible set family of $\mathbf{L}$ over $\vartheta$
as 
\begin{align}
 & \Psi\left(\vartheta\right)\triangleq\left\{ \left\{ l_{z}\left(t\right)\right\} _{z\in\mathcal{Z},t\in\mathcal{T}}:\text{(\ref{eq:multi_thp_c_v2}), (\ref{eq:multi_l_c_v2_2})}\right\} .\label{eq:def_Phi}
\end{align}
It is shown in the following proposition that the region $\Psi\left(\vartheta\right)$
is monotonically increasing over $\vartheta$.
\begin{prop}
\label{prop:opt_c_l}(Monotonicity of $\Psi\left(\vartheta\right)$)
For any $0\le\vartheta_{1}<\vartheta_{2}$, we have $\Psi\left(\vartheta_{1}\right)\subseteq\Psi\left(\vartheta_{2}\right)$.
\end{prop}
\begin{proof}
See Appendix~\ref{sec:proof:opt_c_l}.
\end{proof}
The insight of Proposition~\ref{prop:opt_c_l} is that increasing
$\vartheta$ is equivalent to increasing all transmission power levels.
As a result, a time-frequency resource allocation scheme feasible
under a lower $\vartheta$ must also be feasible under a higher $\vartheta$.

Due to the monotonicity property of $\Psi\left(\vartheta\right)$
with respect to $\vartheta$ in Proposition~\ref{prop:opt_c_l},
the optimal solution to problem (\ref{eq:opt_l_v2}) can be found
using a bisection search on $\vartheta$. The objective is to find
a smallest $\vartheta$ with $|\Psi\left(\vartheta\right)|>0$. Note
that the solution is infeasible if $\Psi\left(\vartheta\right)=\varnothing$.

The algorithm is shown in Algorithm~\ref{alg:l_bi}. Specifically,
we repeat to find the set $\Psi\left(\vartheta\right)$ by searching
$\vartheta$ until $\left|\vartheta_{\max}-\vartheta_{\min}\right|\to0$.
The simplex method can be used to check $\Psi\left(\vartheta\right)=\varnothing$
or not and obtain a point in set $\Psi\left(\vartheta\right)$, since
given any $\vartheta$, the region $\Psi\left(\vartheta\right)$ is
constructed by several linear inequality constraints.

\begin{algorithm}
\# Input: $\mathbf{O}$ and $\mathbf{T}$;
\begin{enumerate}
\item Set $\vartheta\leftarrow(\vartheta_{\max}+\vartheta_{\min})/2$, and
check the feasibility of the set $\Psi\left(\vartheta\right)$ using
a simplex method.
\item If $\Psi\left(\vartheta\right)$ is empty, $\vartheta_{\min}\leftarrow\vartheta$;
otherwise, $\vartheta_{\max}\leftarrow\vartheta$.
\item Repeat from step 1) until $\left|\vartheta_{\max}-\vartheta_{\min}\right|\to0$.
\end{enumerate}
\# Output: $\vartheta\leftarrow\vartheta_{\max}$, $\mathbf{L}\in\Psi(\vartheta_{\max})$,
and $\mathbf{P}$ from (\ref{eq:opt_p_multi}).

\caption{\label{alg:l_bi}Bisection resource allocation algorithm.}
\end{algorithm}

Since $\Psi\left(\vartheta\right)$ is increasing over $\vartheta$
according to Proposition~\ref{prop:opt_c_l}, leading to for any
$\vartheta<\vartheta^{*}$, $\Psi\left(\vartheta\right)=\varnothing$,
it naturally leads to the following optimality result.
\begin{prop}
\label{prop:opt_alg_3}(Optimality of Algorithm~\ref{alg:l_bi})
Algorithm~\ref{alg:l_bi} finds the optimal solution $\mathbf{L}^{*}$
to $\mathscr{P}3$ for any $\{\mathbf{O},\mathbf{T}\}$.
\end{prop}

\subsection{Implementation}

The transmission strategy plan algorithm for multiple commodities
are described in Algorithm~\ref{alg:multi-task}. Specially, in each
iteration, the routes and time boundaries for all commodities are
updated based on the allocated time-frequency resource strategy $\mathbf{L}^{(i-1)}$,
then the time-frequency allocation strategy is updated based on the
allocated routes $\mathbf{O}^{(i)}$ and time boundaries $\mathbf{T}^{(i)}$.

\begin{algorithm}
\# Initialization: Set $\mathbf{t}_{z}^{(0)}\leftarrow\{kT/(M-1)\}_{k\in\{0,\cdots,M-1\}}$
for all $z\in\mathcal{Z}$, random $\mathbf{L}$, and $t\in[0,T]$,
and $i\leftarrow1$.
\begin{enumerate}
\item Obtain $\mathbf{o}_{z}^{(i)}$ and $\mathbf{t}_{z}^{(i)}$ by steps
1) to 3) in Algorithm~\ref{alg:hybrid_opt_alg} based on $\mathbf{t}_{z}^{(i-1)}$
for all $z\in\mathcal{Z}$.
\item Obtain $\mathbf{L}^{(i)}$ by Algorithm~\ref{alg:l_bi} based on
$\mathbf{O}^{(i)}$ and $\mathbf{T}^{(i)}$.
\item Repeat from step 1) until $\left|\mathbf{L}^{(i)}-\mathbf{L}^{(i-1)}\right|\to0$.
\end{enumerate}
\# Output: $\mathbf{O}\leftarrow\{\mathbf{o}_{z}^{(i)}\}_{z\in\mathcal{Z}}$,
$\mathbf{T}\leftarrow\{\mathbf{t}_{z}^{(i)}\}_{z\in\mathcal{Z}}$,
$\mathbf{L}^{(i)}$, and $\mathbf{P}$ from (\ref{eq:opt_p_multi}).

\caption{\label{alg:multi-task}Graph-based multi-commodity transportation
algorithm.}
\end{algorithm}

\begin{figure}
\begin{centering}
\includegraphics[width=1\columnwidth]{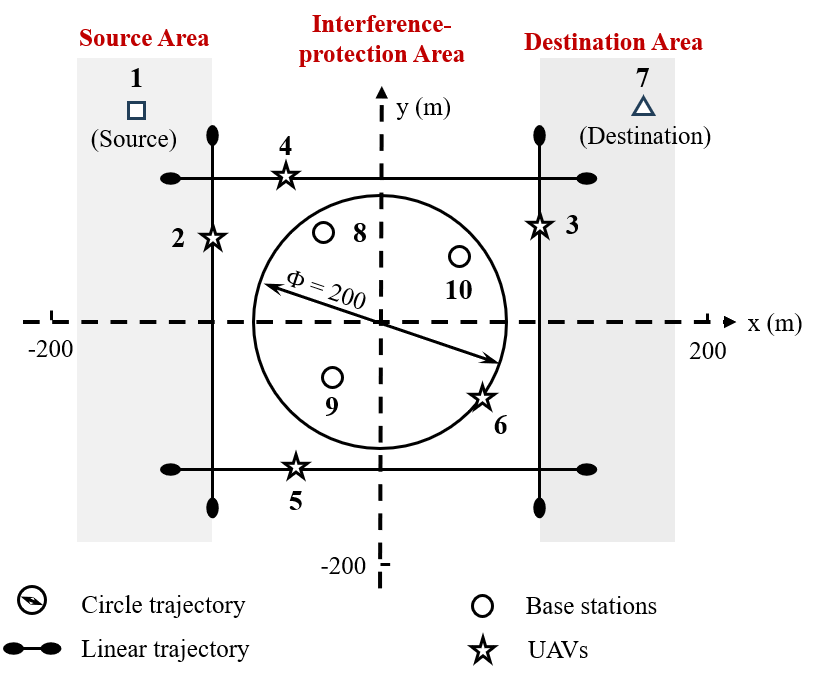}
\par\end{centering}
\caption{\label{fig:Ill_aerial_ad_hoc}Illustration of a sample aerial network
for data delivery, with $5$ aerial nodes (indexed 2 to 6), $3$ neighbor
nodes (indexed 8 to 10), and $1$ source-destination pair (indexed
1 and 7). The source node, neighbor nodes, and destination node are
uniformly randomly placed in the source area, interference-protection
area, and destination area, respectively. The aerial nodes follow
predefined trajectories.}
\end{figure}

\subsubsection{Convergence analysis}

Algorithm~\ref{alg:multi-task} modifies the objective value $\vartheta$
in two parts: updating routes and time boundaries in step 2), where
the value $\vartheta$ decreases over the iteration is shown in Section
\ref{subsec:Convergence-analysis} and updating time-frequency allocation
strategy in step 3), where the value $\vartheta$ decreases over the
iteration because $\mathbf{L}^{(i-1)}$ is a feasible point in $(i)$th
problem and $\mathbf{L}^{(i)}$ is the optimal solution to $\mathscr{P}3$
as shown in Proposition~\ref{prop:opt_alg_3}. As a result, the value
$\vartheta$ decreases over the iteration. It follows that $\vartheta$
is bounded below by $0$, then the Algorithm~\ref{alg:multi-task}
must converge.

\subsubsection{Complexity analysis}

The overall computational complexity is $\mathcal{O}((M^{2}(M+T)Z+(ZT)^{\varphi}\log(ZT))\omega)$,
where $M^{2}(M+T)$ is the complexity of updating route and time boundaries
for single commodity, similar to Section \ref{subsec:complexity_analysis_single},
and $(ZT)^{\varphi}\log(ZT)$ is the complexity of solving the linear
programming problem using simplex algorithm, $\varphi\approx2.38$
\cite{Jan:C20}, and $\omega$ is the iteration number of multi-commodity
algorithm.

\section{Simulation\label{sec:Simulation}}

\begin{figure*}
\begin{centering}
\includegraphics[width=1\textwidth]{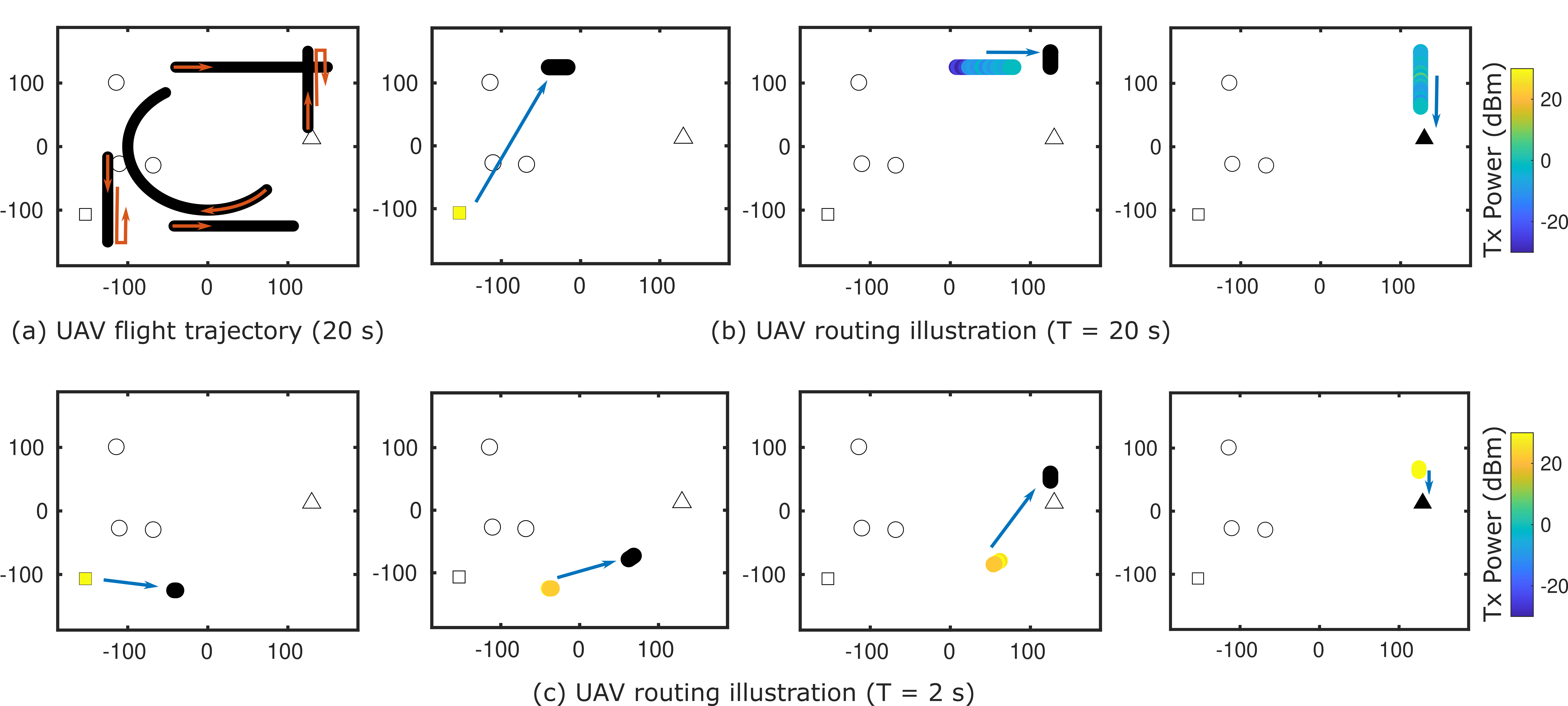}
\par\end{centering}
\caption{\label{fig:case_example}Illustration of two routing results. (a)
The red arrow indicates the initial \ac{uav} flight direction, while
the return arrow indicates that the \ac{uav} returns after reaching
its endpoint. (b) Routing result for a delay-tolerant transmission
task with $T=20\text{ s}$, where 2 relays are selected. (c) Routing
result for a delay-sensitive transmission task with $T=2\text{ s}$,
where 3 relays are selected. In both (b) and (c), transmission power
along the \ac{uav} trajectory is illustrated using a colorbar, and
receiver positions are shown as black circles.}
\end{figure*}

\begin{center}
\begin{table}
\caption{\label{tab:Implementation settings}Default Implementation Parameters}

\centering{}%
\begin{tabular}{>{\raggedright}p{0.3\columnwidth}|>{\raggedright}p{0.6\linewidth}}
\hline 
\textbf{Parameter} & \textbf{Description}\tabularnewline
\hline 
\rowcolor{lightcyan}Cargo \ac{uav} trajectory & Linear trajectory: vertical at 45\,m and horizontal at 50\,m.\tabularnewline
Cargo \ac{uav} hover time & Uniform in $[0,2]$ s.\tabularnewline
\rowcolor{lightcyan}Patrol \ac{uav} trajectory & Circular trajectory at 50\,m altitude.\tabularnewline
\ac{uav} speed & Uniform in $[5,20]$ m/s.\tabularnewline
\rowcolor{lightcyan}Base station location & Uniformly distributed in the interference protection area at ground
level (0\,m).\tabularnewline
Source location & Uniformly distributed in the source area at ground level (0\,m).\tabularnewline
\rowcolor{lightcyan}Destination location & Uniformly distributed in the destination area at ground level (0\,m).\tabularnewline
Carrier frequency & $f_{\text{c}}=3\text{ GHz}$.\tabularnewline
\rowcolor{lightcyan}Bandwidth & $B=10\text{ MHz}$.\tabularnewline
Noise power & $\sigma^{2}=-90\text{ dBm}$.\tabularnewline
\rowcolor{lightcyan}Path loss (LOS) & $22.0+28.0\log_{10}(d)+20\log_{10}(f_{\text{c}})$.\tabularnewline
Path loss (NLOS) & $22.7+36.7\log_{10}(d)+26\log_{10}(f_{\text{c}})$.\tabularnewline
\rowcolor{lightcyan}LOS probability & $\mathbb{P}(\text{LOS},\theta)=(1+6\times\text{exp}(-0.15[\theta-6])^{-1}$,
where $\theta$ is elevation angle.\tabularnewline
Shadowing & Log-normal distribution with 0\,dB mean, 8\,dB variance, and 5\,m
correlation distance.\tabularnewline
\rowcolor{lightcyan}Backtracking factor & $\alpha=0.5$.\tabularnewline
\hline 
\end{tabular}
\end{table}
\par\end{center}

Consider a \ac{uav}-based cargo delivery and patrol system, as shown
in Fig.~\ref{fig:Ill_aerial_ad_hoc}, consisting of four cargo \acpl{uav}
with fixed endpoints (representing merchant and user locations) and
one patrol \ac{uav} following a circular trajectory. The cargo \acpl{uav}
may hover at their endpoints for a random duration, modeled by a uniform
distribution $\mathcal{U}\left(0,2\right)$ s, to simulate dispatch
delays on the merchant side and loading/unloading times on the user
side. The velocity of each aerial node is randomly selected from the
range 5\textendash 20\,m/s, and the initial positions are randomly
assigned to capture variability in order arrivals, cargo-dependent
speed variations, and stochastic data transmission demands.

The source and destination are randomly located within the source
and destination areas on opposite sides, respectively, while the interference-protection
area with random-position \acpl{bs} is in the middle. The altitudes
of sources, vertical \acpl{uav}, horizontal \acpl{uav}, destinations,
and \acpl{bs} are 0 m, 45 m, 50 m, 0 m, and 5 m. The data generated
from the source node will be relayed to the destination node through
the aerial nodes with limited interference to the neighbor network.

The channel gains are realized by $h_{m,n}=g_{m,n}\xi_{m,n}$ according
to (\ref{eq:channel_model}). Specifically, the shape parameters $\kappa_{m,n}$
of Gamma distribution of small-scale fading $\xi_{m,n}$ for air-to-ground
links are set randomly in $[0,30]$, and that for air-to-air channels
are set randomly in $[30,60]$. Same as \cite{LiChe:J24}, the large-scale
fading $g_{m,n}$ includes path loss and shadowing, where the path
loss is generated by 3GPP \ac{umi} model and the channel block state
is generated by \ac{los} probability model, while the shadowing is
modeled by a log-normal distribution, with zero mean and a variance
of 8, and a correlation distance of 5 m. The default implementation
parameters are listed in Table \ref{tab:Implementation settings}.

We compare our performance with the following baselines.

\subsubsection{Aggregate routing \cite{LiChe:J24b}+\cite{AbdNiAboLi:J21}}

This scheme utilizes the predictive channel information to select
the route by the extended Dijkstra\textquoteright s algorithm \cite{AbdNiAboLi:J21}
and optimize the time boundary following the method similar to \cite{LiChe:J24b}.
Specifically, the scheme first constructs the average capacity matrix
$\bar{\mathbf{C}}_{M\times M}$ over the period $[0,T]$, then selects
the route that minimizes $\sum_{k=1}^{K}(K-1)/(\bar{\mathbf{C}}[o_{k},o_{k+1}])$,
where $K$ is the route length, $\bar{\mathbf{C}}[o_{k},o_{k+1}]$
is the average capacity between node $o_{k}$ and $o_{k+1}$, and
$(K-1)$ is the number of relays, indicating the number of segments
for the entire time $T$. Second, the scheme adjusts the time boundaries
using Algorithm~\ref{alg:ht_opt_bi_alg} according to the selected
route.

\subsubsection{Space-time routing \cite{HanXuZhaWan:J23}}

This scheme utilizes the predictive channel information to select
the route using a space-time graph model, a special case of the proposed
dynamic space-time graph, with fixed time boundaries, similar to the
approach in \cite{HanXuZhaWan:J23}, but without time boundary optimization.
Specifically, the scheme constructs the space-time graph as Section
\ref{subsec:graph_model} by setting $\mathbf{t}=\{kT/(M-1)\}_{k\in\{0,\cdots,M-1\}}$
and select the route using bottleneck path planning algorithm.

\subsubsection{Brute-force (Optimum)}

This scheme enumerates all possible paths from source to destination
and calculates the corresponding optimal time boundaries using Algorithm~\ref{alg:ht_opt_bi_alg},
then selects the path with least leakage power as the transmission
route. Its solution is optimal because Algorithm~\ref{alg:ht_opt_bi_alg}
can find the optimal time boundaries to $\mathscr{P}2$ for each route
according to Proposition~\ref{prop:opt_alg_1}.

\subsection{Single Commodity Performance}

To illustrate the operation of the proposed routing scheme, we provide
two representative routing examples in Figure \ref{fig:case_example}.
These cases demonstrate how the proposed method flexibly leverages
\ac{uav} mobility and multi-hop relaying to create spatial proximity
for effective data transmission. Specifically, delay-tolerant tasks
can exploit the mobility-induced proximity of \acpl{uav} as shown
in Subfigure (b), while delay-sensitive tasks rely on relays to artificially
create proximity Subfigure (c).

Fig.~\ref{fig:qos_c} illustrates the ratio of actual transmitted
data (throughput) to the planned data size (commodity size), where
a ratio of $1$ indicates that the transmitted data fully matches
the planned amount. The results show that the median values are all
greater than 1, which indicates the \ac{qos} constraint is satisfied
in 100\% of cases on average and confirms that the throughput constraint
is met in the expected sense as formulated in the problem.

We then demonstrate the near-optimality of the single commodity algorithm.
Fig.~\ref{fig:cdf_itfp_rand} shows the \ac{cdf} of interference
power under the random \ac{uav} initial positions, tolerable time
$T\in[1,60]$ s, data size $S\in[5,500]$ Mbits, neighbor network
size $M\in[1,20]$. The results show that the performance of the proposed
scheme is almost identical to the solution obtained via the brute-force
algorithm, confirming the optimality of the proposed algorithm. Furthermore,
the proposed scheme outperforms baseline schemes by approximately
13 dB, indicating that it can reduce interference to neighboring networks
by more than 10 times.

\begin{figure}
\begin{centering}
\includegraphics[width=1\columnwidth]{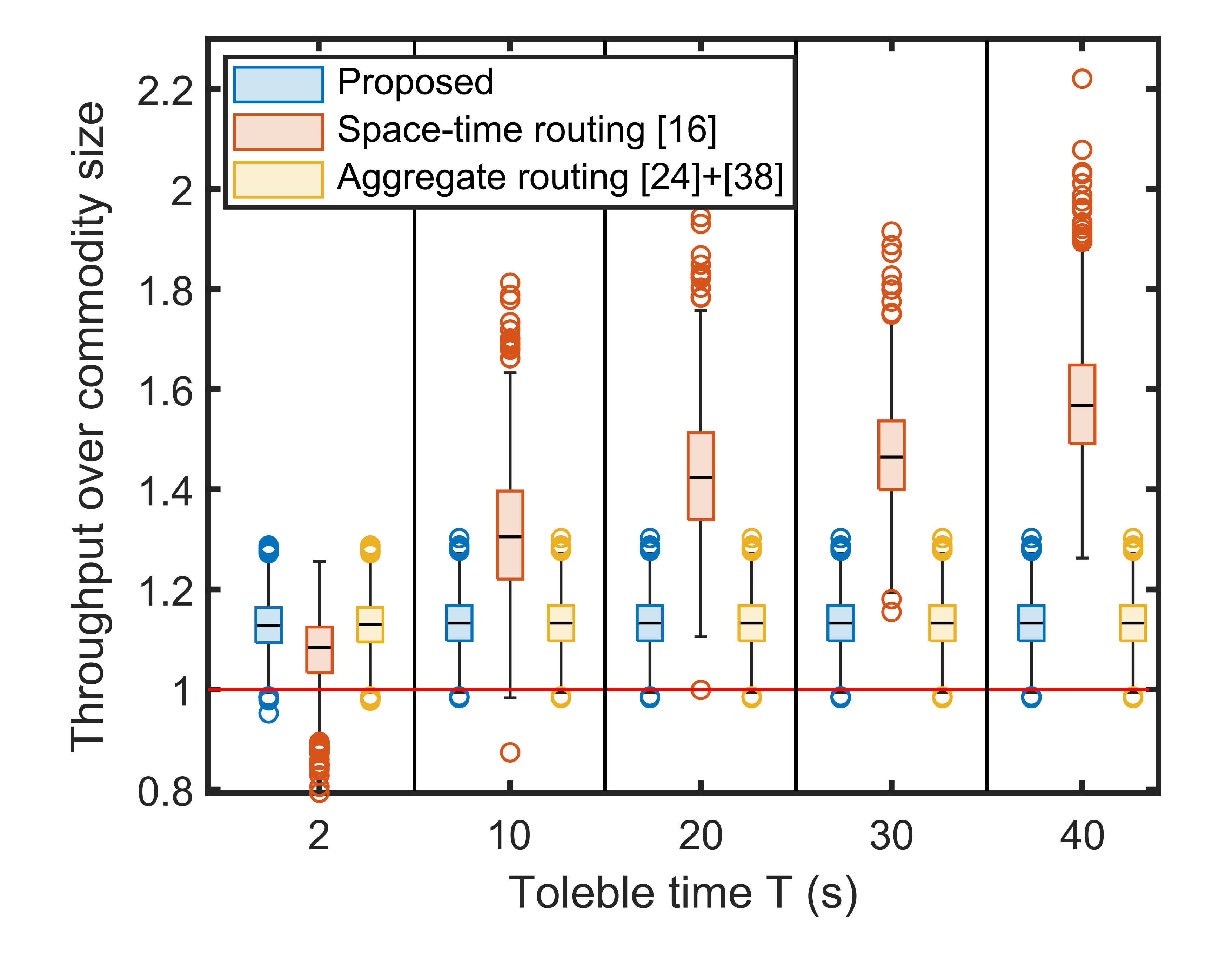}
\par\end{centering}
\caption{\label{fig:qos_c}\Ac{qos} satisfaction under different tolerable
transmission times.}
\end{figure}

\begin{figure}
\begin{centering}
\includegraphics[width=1\columnwidth]{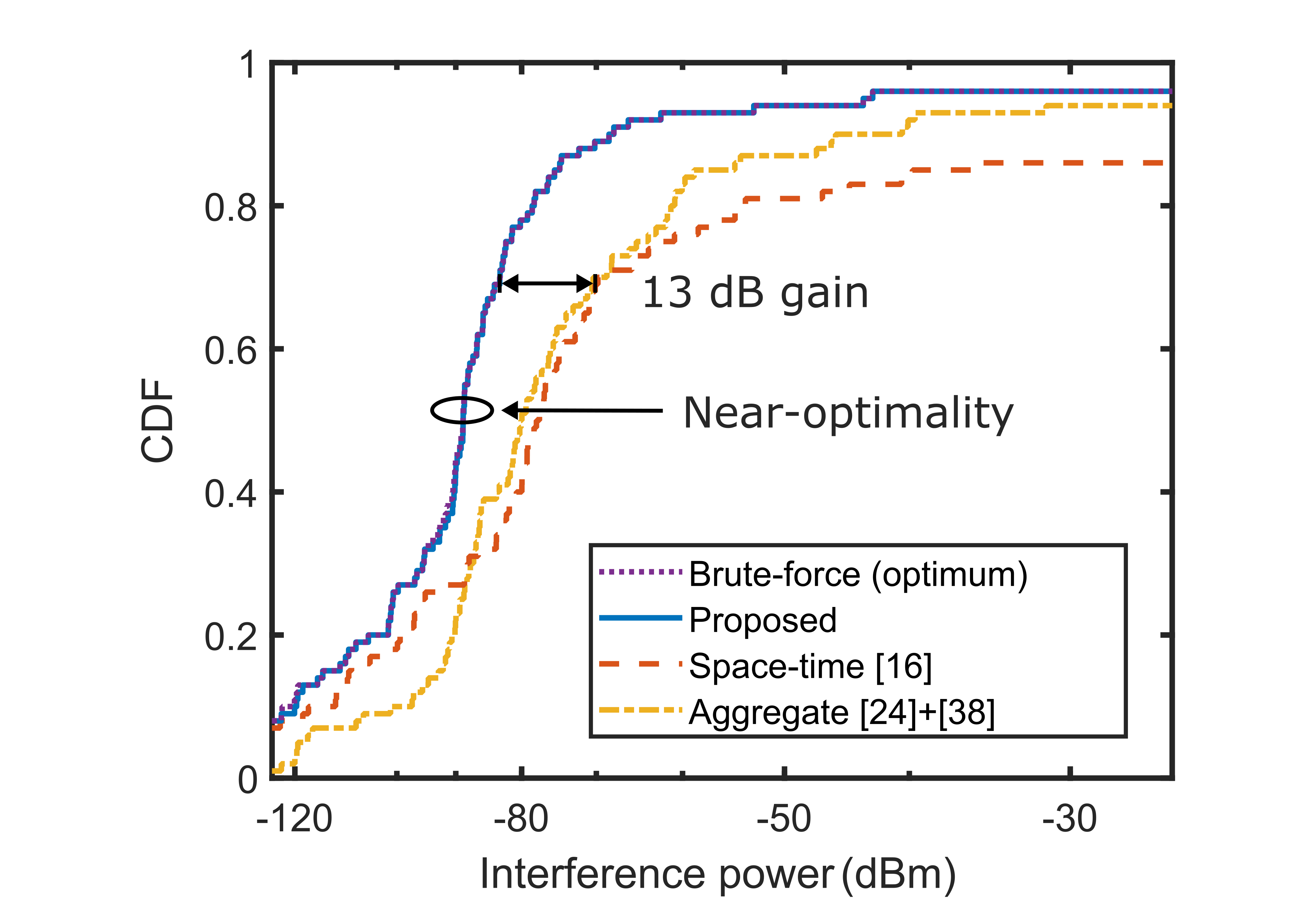}
\par\end{centering}
\caption{\label{fig:cdf_itfp_rand}The \ac{cdf} of the interference leakage
power $\vartheta$ of $100$ replicated random experiments, including
random initial \ac{uav} positions, speeds, and hovering times; random
source, destination, and neighboring node locations; and random data
sizes.}
\end{figure}

Fig.~\ref{fig:itfp_over_T} shows the interference leakage power
under different tolerable time $T$ and data size $S$. The results
demonstrate that the performance gain is particularly significant
for delay-sensitive scenarios (small $T$) and large data sizes (large
$S$). On average, the proposed scheme achieves a $6$ dB and $14$
dB improvement over the space-time routing scheme and the aggregate
routing scheme, respectively. Furthermore, both baseline schemes result
in over 25 dB more interference than the proposed algorithm when $T=1\text{ s}$
and $S=50\text{ Mbits}$. This demonstrates that, in delay-sensitive
scenarios involving large-volume data transmissions, both route selection
and time-boundary optimization are essential, the absence of either
leads to substantially higher interference.

\begin{figure}
\begin{centering}
\includegraphics[width=1\columnwidth]{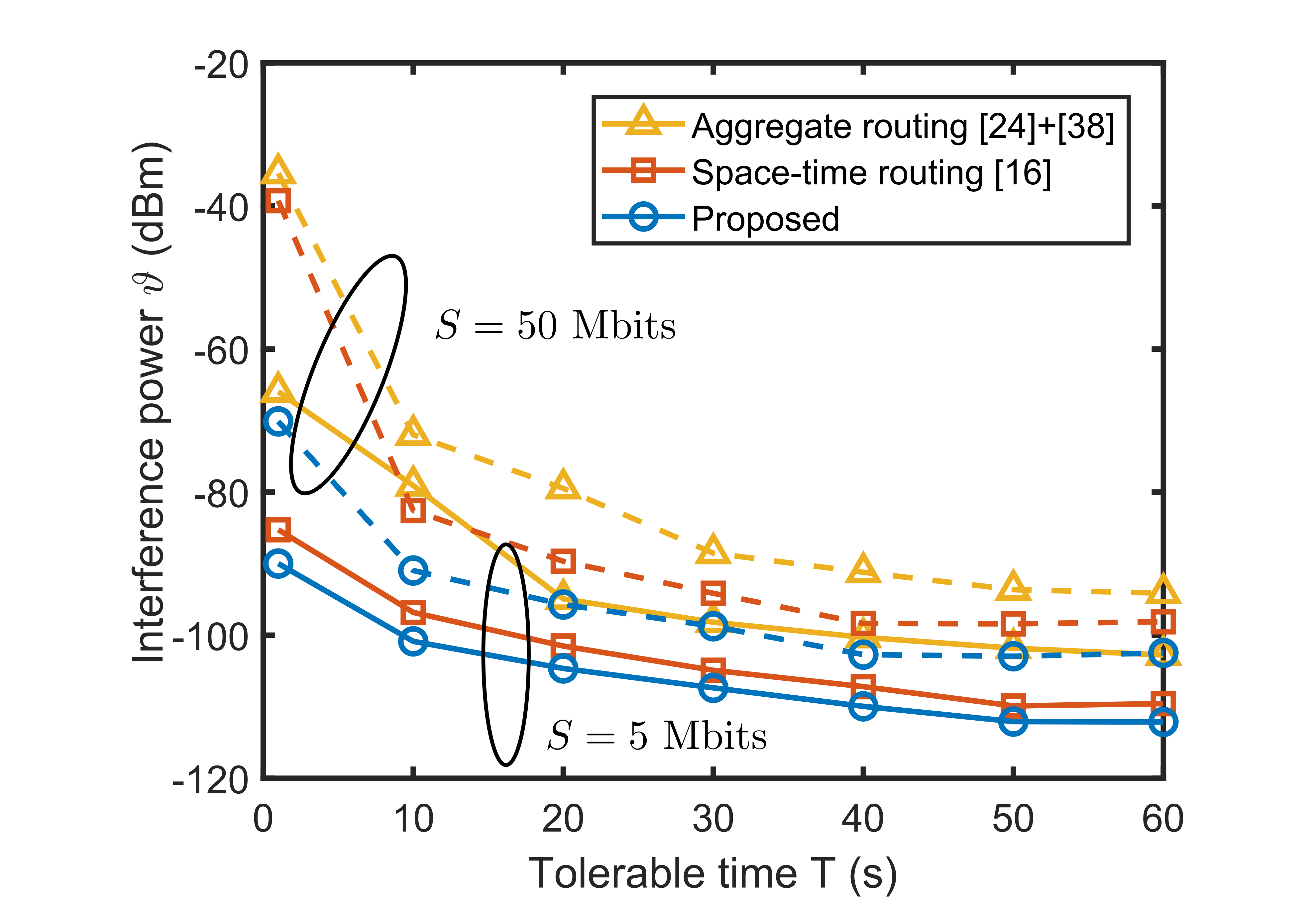}
\par\end{centering}
\caption{\label{fig:itfp_over_T}The interference leakage power under different
tolerable time $T$, evaluated for two data sizes: $S=5\text{ Mbits}$
and $S=50\text{ Mbits}$.}
\end{figure}
\begin{figure}
\begin{centering}
\includegraphics[width=1\columnwidth]{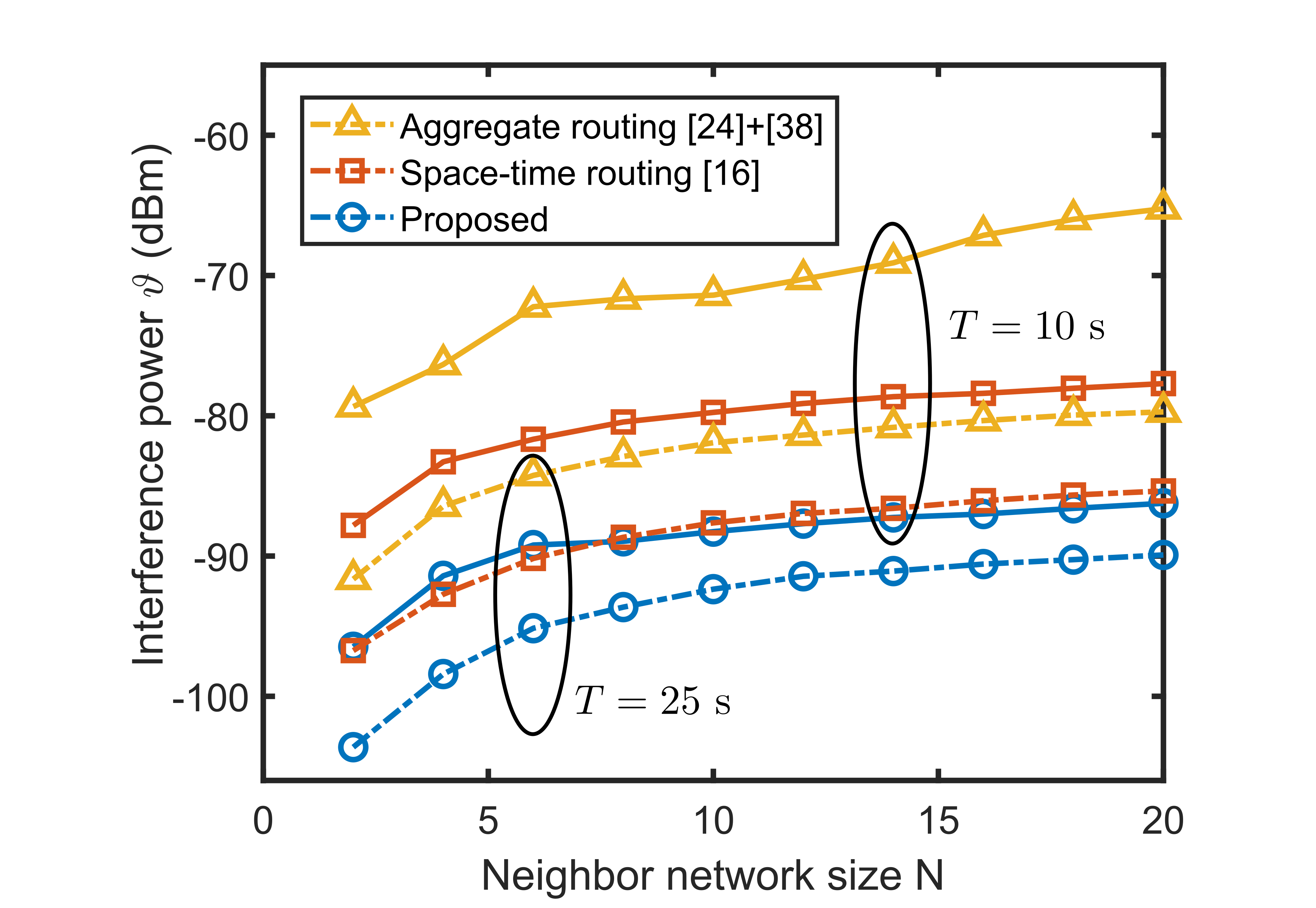}
\par\end{centering}
\caption{\label{fig:itfp_over_N}The interference leakage power under different
neighbor network size $N$, evaluated for two tolerable time: $T=10\text{ s}$
and $T=25\text{ s}$.}
\end{figure}

Fig.~\ref{fig:itfp_over_N} shows the interference leakage power
under different neighbor network size $N$, highlighting the robustness
of the proposed algorithm across various environments. While the interference
power increases with the density of the neighbor network, the proposed
scheme consistently achieves a performance gain of more than 4 dB
and 14 dB compared to the baselines. In addition, the interference
leakage power of the proposed scheme for $T=10$ s is almost lower
than that of the baselines for $T=25$ s, demonstrating that even
when the available time is reduced by half, the proposed scheme still
generates less interference leakage to the neighbor network.

\subsection{Multi-commodity Performance}

In this section, we evaluate the proposed algorithm in a multi-commodity
scenario and demonstrate the performance improvement achieved by splitting
a single large data into multiple smaller data for transmission.

Fig.~\ref{fig:itfp_over_Z} shows the interference leakage power
under different commodity number $Z$, highlighting the significant
improvements achieved in dense service scenarios (large $Z$). The
results show that the performance gap between the proposed scheme
and the baselines widens as the number of commodities increases. Specifically,
as $Z$ increases from $1$ to $19$, the interference leakage power
caused by the baselines rises by approximately $30$ dB, while that
caused by the proposed scheme only increases by about $10$ dB. This
indicates that the performance gap grows by a factor of $100$ as
the number of commodities increases by $20$ times. In other words,
for every 1X increase in the number of commodities, the performance
gap with the classical methods grows by a factor of $5$.

\begin{figure}
\begin{centering}
\includegraphics[width=1\columnwidth]{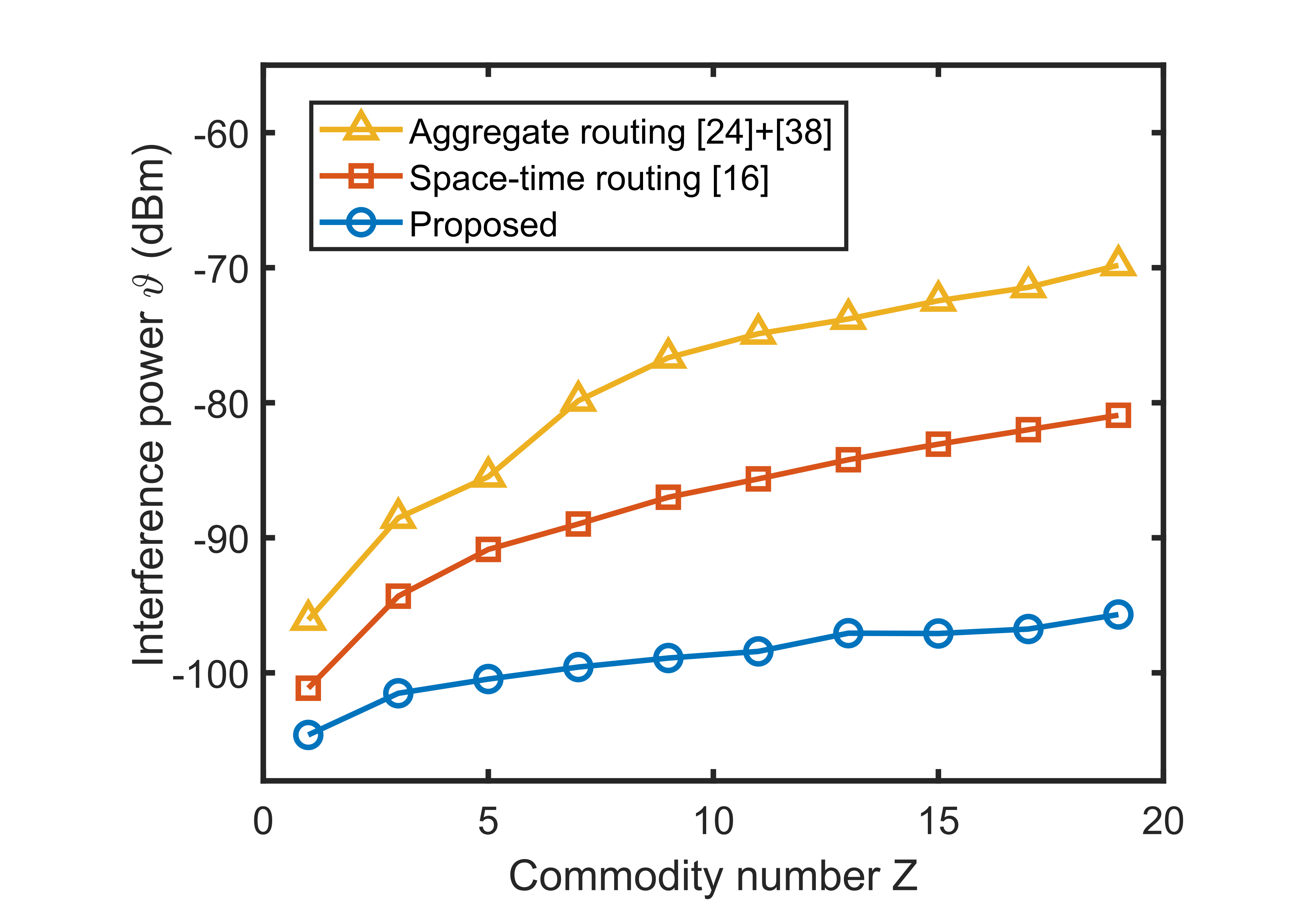}
\par\end{centering}
\caption{\label{fig:itfp_over_Z}The interference leakage power under different
commodity number $Z$. Here, each flow has a different source-destination
pair and the data size of each commodity is $S_{z}=S$.}
\end{figure}

\begin{figure}
\begin{centering}
\includegraphics[width=1\columnwidth]{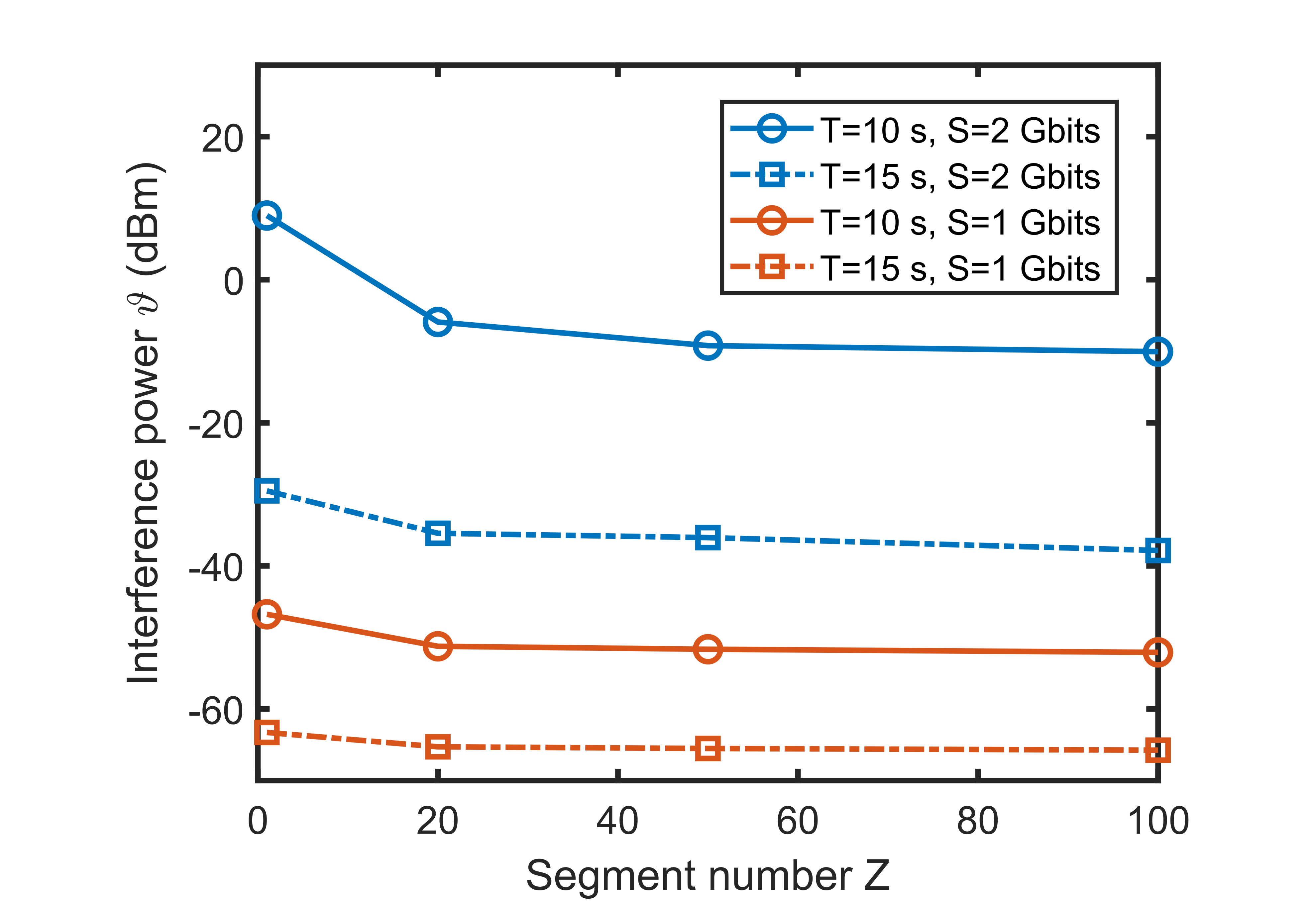}
\par\end{centering}
\caption{\label{fig:itfp_over_seg}The interference leakage power under different
segment number $Z$. Here, all flows share the same source-destination
pair and the data size of each commodity is $S_{z}=S/Z$.}
\end{figure}

Fig.~\ref{fig:itfp_over_seg} shows the interference leakage power
under different segment number $Z$, verifying that splitting a single
large data into multiple smaller segments for transmission improves
performance. Overall, the interference leakage power decreases as
the segment number increases, with significant improvements ($19$
dB) observed in delay-sensitive ($T=10$ s) and large-data ($S=2$
Gbits) transmissions. This is because fine-grained resource allocation
becomes critical when resources are relatively scarce, such as in
cases of small $T$ and large $S$.

\section{Conclusion\label{sec:Conclusion}}

This work addressed network-level optimization in low-altitude aerial
networks and proposed a dynamic space-time graph with virtual edges
and developed a cross-layer optimization framework to decouple resource
allocation from routing. Then an efficient single-commodity transportation
algorithm is developed by analyzing the optimality and deriving deterministic
capacity lower bound. Simulation results show that the proposed single-commodity
algorithm is almost optimal in various cases and achieves $30$ dB
performance improvement over the classical methods for delay-sensitive
and large data transportation tasks. In addition, the algorithm is
extended for multi-commodity transportation, and an efficient time-frequency
allocation algorithm is proposed. Simulation results show that the
proposed method achieves $100$X improvements in dense service scenarios,
and for a single large commodity, segmenting it into smaller parts
for transmission further achieves an additional 20 dB performance
improvement. As a result, the proposed strategy is more suitable for
deployment in spectrum-sharing low altitude networks.

\appendices

\section{Proof of Proposition~\ref{prop:power_allocation_policy}\label{sec:proof_power_allocation_policy}}

According to the conditions (\ref{eq:KKT_c1}), (\ref{eq:KKT_c4}),
and (\ref{eq:KKT_c5}), we prove that all optimal Lagrangian parameters
are positive. First, there must exist $t\in[t_{k},t_{k+1})$ such
that $v(t)>0$; otherwise, $v(t)\equiv0$ for all $t$ according to
(\ref{eq:KKT_c1}) and $\partial L/\partial\vartheta_{m,n}=1\neq0$,
contradicting the condition (\ref{eq:KKT_c4}). Then, according to
the condition (\ref{eq:KKT_c5}), we have $\mu>0$ because $\forall t$
\[
\mu=v\left(t\right)\max_{j\in\mathcal{N}}\left\{ h_{m,j}\left(t\right)\right\} /\mathbb{E}\left[\frac{B}{\ln2}\frac{h_{m,n}\left(t\right)}{\delta^{2}+p_{m,n}\left(t\right)h_{m,n}\left(t\right)}\right]
\]
which is greater than $0$ due to $v(t)>0,\exists t$. Therefore,
it can be derived that $v(t)>0$, $\forall t\in[t_{k},t_{k+1})$ due
to 
\[
v\left(t\right)=\frac{\mu\mathbb{E}\left[\frac{B}{\ln2}\frac{h_{m,n}\left(t\right)}{\delta^{2}+p_{m,n}\left(t\right)h_{m,n}\left(t\right)}\right]}{\max_{j\in\mathcal{N}}\left\{ h_{m,j}\left(t\right)\right\} }
\]
 according to condition (\ref{eq:KKT_c5}) and $\mu>0$.

Then, according to the conditions (\ref{eq:KKT_c2}) and (\ref{eq:KKT_c3}),
and the positive Lagrangian parameters, there must be 
\[
\max_{j\in\mathcal{N}}\left\{ p_{m,n}\left(t\right)h_{m,j}\left(t\right)\right\} -\vartheta_{m,n}=0,\forall t
\]
and
\begin{equation}
S-\int_{t_{k}}^{t_{k+1}}\mathbb{E}\left[B\log_{2}\left(1+\frac{p_{m,n}\left(t\right)h_{m,n}\left(t\right)}{\delta^{2}}\right)\right]dt=0.\label{eq:kkt_thp_c}
\end{equation}

Therefore, the optimal transmission power policy is 
\begin{equation}
p_{m,n}^{*}\left(t\right)=\vartheta_{m,n}^{*}/\max_{j\in\mathcal{N}}\left\{ h_{m,j}\left(t\right)\right\} .\label{eq:opt_p}
\end{equation}
Substitute (\ref{eq:opt_p}) into (\ref{eq:kkt_thp_c}), the optimal
interference power leakage $\vartheta_{m,n}^{*}$ is the solution
to 
\[
\int_{t_{k}}^{t_{k+1}}\mathbb{E}\left[B\log_{2}\left(1+\frac{\vartheta_{m,n}h_{m,n}\left(t\right)}{\max_{j\in\mathcal{N}}\left\{ h_{m,j}\left(t\right)\right\} \delta^{2}}\right)\right]dt=S.
\]

\section{Proof of Proposition~\ref{prop:lb_c}\label{sec:proof:lb_c}}

The the expected capacity is lower bounded by 
\begin{multline*}
\mathbb{E}\left[c_{m,n}\left(t\right)\right]=\mathbb{E}\left[\log_{2}\left(1+\frac{\vartheta h_{m,n}\left(t\right)}{\max_{j\in\mathcal{N}}\left\{ h_{m,j}\left(t\right)\right\} \delta^{2}}\right)\right]\\
\ge\mathbb{E}\left[\log_{2}\left(1+\frac{\vartheta h_{m,n}\left(t\right)}{\mathbb{E}\left[\max_{j\in\mathcal{N}}\left\{ h_{m,j}\left(t\right)\right\} \right]\delta^{2}}\right)\right]
\end{multline*}
based on the Jensen's inequality and the convexity of function $\log_{2}(1+a/x)$.
Denote $Y_{m,j}(t)\triangleq h_{m,j}(t)-\mathbb{E}[h_{m,j}(t)]$,
then we have 
\begin{align*}
\max_{j\in\mathcal{N}}\left\{ h_{m,j}\left(t\right)\right\}  & =\max_{j\in\mathcal{N}}\left\{ Y_{m,j}\left(t\right)+\mathbb{E}\left[h_{m,j}\left(t\right)\right]\right\} \\
 & \le\max_{j\in\mathcal{N}}\left\{ Y_{m,j}\left(t\right)\right\} +\max_{j\in\mathcal{N}}\left\{ \mathbb{E}\left[h_{m,j}\left(t\right)\right]\right\} .
\end{align*}
Take the expectation on both side, we have 
\begin{equation}
\mathbb{E}\left[\max_{j\in\mathcal{N}}\left\{ h_{m,j}\left(t\right)\right\} \right]\le\mathbb{E}\left[\max_{j\in\mathcal{N}}\left\{ Y_{m,j}\left(t\right)\right\} \right]+\max_{j\in\mathcal{N}}\left\{ \mathbb{E}\left[h_{m,j}\left(t\right)\right]\right\} .\label{eq:lb_m_1}
\end{equation}
Denote $\mathbb{V}[h_{m,j}(t)]$ as the variance of $h_{m,j}(t)$,
we have
\begin{align*}
 & \left(\mathbb{E}\left[\max_{j\in\mathcal{N}}\left\{ Y_{m,j}\left(t\right)\right\} \right]\right)^{2}\stackrel{(a)}{\le}\mathbb{E}\left[\left(\max_{j\in\mathcal{N}}\left\{ Y_{m,j}\left(t\right)\right\} \right)^{2}\right]\\
 & \quad\stackrel{(b)}{\le}\mathbb{E}\left[\max_{j\in\mathcal{N}}\left\{ Y_{m,j}\left(t\right)^{2}\right\} \right]\stackrel{(c)}{\le}\mathbb{E}\left[\sum_{j\in\mathcal{N}}Y_{m,j}\left(t\right)^{2}\right]\\
 & \quad\stackrel{(d)}{=}\sum_{j\in\mathcal{N}}\mathbb{E}\left[\left(h_{m,j}\left(t\right)-\mathbb{E}\left[h_{m,j}\left(t\right)\right]\right)^{2}\right]=\sum_{j\in\mathcal{N}}\mathbb{V}\left[h_{m,j}\left(t\right)\right]
\end{align*}
where (a) holds because of the Jensen's inequality and the convexity
of function $x^{2}$, (b) holds because of the convexity of the function
$\max(x)$, (c) holds because of $\max_{j\in\mathcal{N}}\{Y_{m,j}(t)^{2}\}\le\sum_{j\in\mathcal{N}}Y_{m,j}(t)^{2}$,
and (d) holds based on the definition of $Y_{m,j}\left(t\right)$.
Then, (\ref{eq:lb_m_1}) becomes 
\begin{align*}
 & \mathbb{E}\left[\max_{j\in\mathcal{N}}\left\{ h_{m,j}\left(t\right)\right\} \right]\le\max_{j\in\mathcal{N}}\left\{ \mathbb{E}\left[h_{m,j}\left(t\right)\right]\right\} +\sqrt{\sum_{j\in\mathcal{N}}\mathbb{V}\left[h_{m,j}\left(t\right)\right]}\\
 & \quad=\max_{j\in\mathcal{N}}\left\{ g_{m,j}\left(t\right)\right\} +\sqrt{\sum_{j\in\mathcal{N}}g_{m,j}\left(t\right)^{2}/\kappa_{m,j}\left(t\right)}
\end{align*}
where the equality holds because $h_{m,j}\left(t\right)$ follows
$\text{Gamma}(\kappa_{m,j}(t),g_{m,j}(t)/\kappa_{m,j}(t))$. Denote
$\omega_{m}(t)\triangleq\sqrt{\sum_{j\in\mathcal{N}}g_{m,j}(t)^{2}/\kappa_{m,j}(t)}$,
then, the expected capacity function $\mathbb{E}\{c_{m,n}(t)\}$ is
lower bounded by 
\begin{align}
 & \mathbb{E}\left[\log_{2}\left(1+\frac{\vartheta h_{m,n}\left(t\right)}{\left(\max_{j\in\mathcal{N}}\left\{ g_{m,j}\left(t\right)\right\} +\omega_{m}\left(t\right)\right)\delta^{2}}\right)\right]\nonumber \\
 & \stackrel{(a)}{\ge}\log_{2}\left(1+\frac{\vartheta g_{m,n}\left(t\right)}{\left(\max_{j\in\mathcal{N}}\left\{ g_{m,j}\left(t\right)\right\} +\omega_{m}\left(t\right)\right)\delta^{2}}\right)-\epsilon_{m,n}\left(t\right)\label{eq:appx_1_appx_2}
\end{align}
where $\epsilon_{m,n}(t)=\log_{2}(e)/\kappa_{m,n}(t)-\log_{2}(1+(2\kappa_{m,n}(t))^{-1})$,
and (a) holds according to \cite[Lemma 1]{LiChe:J24b}. In addition,
when targeted channel $h_{m,n}\left(t\right)$ is \ac{los}, that
is, $\kappa_{m,n}(t)\to\infty$, the both sides of (\ref{eq:appx_1_appx_2})
are equal.

\section{Proof of Proposition~\ref{prop:opt_c_p2_over_o}\label{sec:proof_opt_c_p2_over_o}}

Before proving the optimality to $\mathscr{P}2$, we derive a necessary
lemma first.
\begin{lem}
\label{lem:monotonicity}(Monotonicity of $w_{m,n}^{k}(t_{k},t_{k+1})$)
For any $t_{k}<t_{k+1}^{\prime}<t_{k+1}^{\prime\prime}$, it holds
that $w_{m,n}^{k}(t_{k},t_{k+1}^{\prime})>w_{m,n}^{k}(t_{k},t_{k+1}^{\prime\prime})$.
For any $t_{k}^{\prime}<t_{k}^{\prime\prime}<t_{k+1}$, it holds that
$w_{m,n}^{k}(t_{k}^{\prime},t_{k+1})<w_{m,n}^{k}(t_{k}^{\prime\prime},t_{k+1})$.
\end{lem}
\begin{proof}
Denote $\vartheta^{\prime}=w_{m,n}^{k}(t_{k},t_{k+1}^{\prime})$ and
$\vartheta^{\prime\prime}=w_{m,n}^{k}(t_{k},t_{k+1}^{\prime\prime})$
are the optimal values of problem $\mathscr{P}1$ with time boundaries
$\{t_{k},t_{k+1}^{\prime}\}$ and $\{t_{k},t_{k+1}^{\prime\prime}\}$.
According to optimal condition in Proposition~\ref{prop:power_allocation_policy},
we have 
\begin{align*}
 & \int_{t_{k}}^{t_{k+1}^{\prime}}\mathbb{E}\left[B\log_{2}\left(1+\frac{\vartheta^{\prime}h_{m,n}\left(t\right)}{\max_{j\in\mathcal{N}}\left\{ h_{m,j}\left(t\right)\right\} \delta^{2}}\right)\right]dt\\
 & =\int_{t_{k}}^{t_{k+1}^{\prime\prime}}\mathbb{E}\left[B\log_{2}\left(1+\frac{\vartheta^{\prime\prime}h_{m,n}\left(t\right)}{\max_{j\in\mathcal{N}}\left\{ h_{m,j}\left(t\right)\right\} \delta^{2}}\right)\right]dt
\end{align*}
and equals to $S$. Spiting the time $[t_{k},t_{k}^{\prime\prime})$
to two parts $[t_{k},t_{k}^{\prime})$ and $[t_{k}^{\prime},t_{k}^{\prime\prime})$,
we have 
\begin{align*}
 & \int_{t_{k}}^{t_{k+1}^{\prime}}\mathbb{E}\Bigg[B\log_{2}\left(1+\frac{\vartheta^{\prime}h_{m,n}\left(t\right)}{\max_{j\in\mathcal{N}}\left\{ h_{m,j}\left(t\right)\right\} \delta^{2}}\right)\\
 & \quad\quad\quad\quad\ \,-B\log_{2}\left(1+\frac{\vartheta^{\prime\prime}h_{m,n}\left(t\right)}{\max_{j\in\mathcal{N}}\left\{ h_{m,j}\left(t\right)\right\} \delta^{2}}\right)\Bigg]dt\\
 & =\int_{t_{k+1}^{\prime}}^{t_{k+1}^{\prime\prime}}\mathbb{E}\left[B\log_{2}\left(1+\frac{\vartheta^{\prime\prime}h_{m,n}\left(t\right)}{\max_{j\in\mathcal{N}}\left\{ h_{m,j}\left(t\right)\right\} \delta^{2}}\right)\right]dt\\
 & >0
\end{align*}
which means $\vartheta^{\prime}>\vartheta^{\prime\prime}$. In other
word, $w_{m,n}^{k}(t_{k},t_{k+1}^{\prime})>w_{m,n}^{k}(t_{k},t_{k+1}^{\prime\prime})$.

Similarly, it can be proven that $w_{m,n}^{k}(t_{k}^{\prime},t_{k+1})<w_{m,n}^{k}(t_{k}^{\prime\prime},t_{k+1})$
for any $t_{k}^{\prime}<t_{k}^{\prime\prime}<t_{k+1}$.
\end{proof}
Given route variable $\mathbf{o}$, first, we prove the sufficiency.
Assuming that $\{\mathbf{t}^{*},\vartheta^{*}\}$ is the optimal solution
to $\mathscr{P}2$, then, we will show that $w_{o(k),o(k+1)}^{k}(t_{k}^{*},t_{k+1}^{*})=\vartheta^{*}$
for all $k\in\{1,\cdots,M-1\}$ by contradiction.

Suppose that there is a $k\in\{1,\cdots,M-1\}$ with $w_{o(k),o(k+1)}^{k}(t_{k}^{*},t_{k+1}^{*})\triangleq\vartheta_{k}<\vartheta^{*}$.
Then we can find $t_{k}^{\prime}>t_{k}^{*}$, $t_{k+1}^{\prime}<t_{k+1}^{*}$,
and $\vartheta_{k}^{\prime}\in(\vartheta_{k},\vartheta^{*})$ with
$w_{o(k),o(k+1)}^{k}(t_{k}^{\prime},t_{k+1}^{\prime})=\vartheta_{k}^{\prime}$
due to the monotonicity property of $w_{o(k),o(k+1)}^{k}(t_{k},t_{k+1})$
over $t_{k}$ (increasing) and $t_{k+1}$ (decreasing).

For example, we first determine $t_{k}^{\prime}$ by solving the following
equation 
\[
w_{o(k),o(k+1)}^{k}(t_{k}^{\prime},t_{k+1}^{*})=\vartheta_{k}^{\prime}+(\vartheta_{k}^{\prime}-\vartheta_{k})/2
\]
where $t_{k}^{\prime}>t_{k}^{*}$ because $\vartheta_{k}+(\vartheta_{k}^{\prime}-\vartheta_{k})/2>\vartheta_{k}$
due to $\vartheta_{k}^{\prime}>\vartheta_{k}$ and the increasing
monotonicity of $w_{o(k),o(k+1)}^{k}(t_{k},t_{k+1})$ over $t_{k}$.
Then, we determine $t_{k+1}^{\prime}$ by solving the following equation
\[
w_{o(k),o(k+1)}^{k}(t_{k}^{\prime},t_{k+1}^{\prime})=\vartheta_{k}^{\prime}
\]
where $t_{k+1}^{\prime}<t_{k+1}^{*}$ because $\vartheta_{k}+(\vartheta_{k}^{\prime}-\vartheta_{k})/2<\vartheta_{k}^{\prime}$
due to $\vartheta_{k}^{\prime}>\vartheta_{k}$ and the decreasing
monotonicity of $w_{o(k),o(k+1)}^{k}(t_{k},t_{k+1})$ over $t_{k+1}$.

Then, we have $w_{o(k-1),o(k)}^{k-1}(t_{k-1}^{*},t_{k}^{\prime})<\vartheta^{*}$
and $w_{o(k+1),o(k+2)}^{k+1}(t_{k+1}^{\prime},t_{k+2}^{*})<\vartheta^{*}$
because $t_{k}^{\prime}>t_{k}^{*}$, $t_{k+1}^{\prime}<t_{k+1}^{*}$,
and the monotonicity property of $w_{o(k),o(k+1)}^{k}(t_{k},t_{k+1})$
over $t_{k}$ (increasing) and $t_{k}$ (decreasing). By applying
the same procedure, we can find $t_{k-1}^{\prime}>t_{k-1}^{*}$ and
$\vartheta_{k-1}^{\prime}<\vartheta^{*}$ such that $w_{o(k-1),o(k)}^{k-1}(t_{k-1}^{\prime},t_{k}^{\prime})=\vartheta_{k-1}^{\prime}$,
and $t_{k+1}^{\prime}<t_{k+1}^{*}$ and $\vartheta_{k+1}^{\prime}<\vartheta^{*}$
such that $w_{o(k+1),o(k+2)}^{k+1}(t_{k+1}^{\prime},t_{k+2}^{\prime})=\vartheta_{k+1}^{\prime}$.

By induction, one can construct a new time sequence $\{t_{k}^{\prime}\}_{k\in\{1,\cdots,M\}}$
with $w_{o(k),o(k+1)}^{k}(t_{k}^{\prime},t_{k+1}^{\prime})=\vartheta_{k}^{\prime}<\vartheta^{*}$
for all $k\in\{1,\cdots,M-1\}$. Let $\vartheta^{\prime}=\max_{k\in\{1,\cdots,M-1\}}\{\vartheta_{k}^{\prime}\}$.
We have $\{\mathbf{t}^{\prime},\vartheta^{\prime}\}$ is feasible
for problem $\mathscr{P}2$ and $\vartheta^{\prime}<\vartheta$, which
indicate that $\{\mathbf{t}^{*},\vartheta^{*}\}$ is sub-optimal.
This contradicts the initial assumption.

In other words, $\{\mathbf{t}^{*},\vartheta^{*}\}$ is the optimal
solution to problem $\mathscr{P}2$ only if $w_{o(k),o(k+1)}^{k}(t_{k}^{*},t_{k+1}^{*})=\vartheta^{*}$
for all $k\in\{1,\cdots,M-1\}$.

Next, we prove the necessity. Assuming that $\{\mathbf{t}^{*},\vartheta^{*}\}$
is a point with $w_{o(k),o(k+1)}^{k}(t_{k}^{*},t_{k+1}^{*})=\vartheta^{*}$
for all $k\in\{1,\cdots,M-1\}$, then, we will show that $\{\mathbf{t}^{*},\vartheta^{*}\}$
is the optimal solution to problem $\mathscr{P}2$ by contradiction.

Suppose that $\{\mathbf{t}^{*},\vartheta^{*}\}$ is not the optimal
solution to problem $\mathscr{P}2$, that is, there is a feasible
point for $\mathscr{P}2$, $\{\mathbf{t}^{\prime},\vartheta^{\prime}\}$,
with $\vartheta^{\prime}<\vartheta^{*}$. To ensure that the constraints
$w_{o(k),o(k+1)}^{k}(t_{k},t_{k+1})\le\vartheta$ holding for $k\in\{1,2,\cdots,M-2\}$,
there must be $t_{k}^{\prime}>t_{k}^{*}$ for $k\in\{2,\cdots,M-1\}$
according to Appendix~\ref{sec:proof_c_increasing_over_vartheta}.
However, the constraint $w_{o(k),o(k+1)}^{k}(t_{k},t_{k+1})\le\vartheta$
for $k=M-1$ cannot hold because 
\begin{align*}
w_{o(M-1),o(M)}^{M-1}\left(t_{M-1}^{\prime},t_{M}^{\prime}\right) & \stackrel{(a)}{\ge}w_{o(M-1),o(M)}^{M-1}\left(t_{M-1}^{*},t_{M}^{*}\right)\\
 & =\vartheta^{*}\stackrel{(b)}{>}\vartheta^{\prime}
\end{align*}
where $(a)$ holds because $t_{M-1}^{\prime}>t_{M-1}^{*}$, $t_{M}^{\prime}=t_{M}^{*}=T$,
and the increasing monotonicity of $w_{o(k),o(k+1)}^{k}(t_{k},t_{k+1})$
over $t_{k}$; (b) holds because of the hypothesis. As a result, the
point $\{\mathbf{t}^{\prime},\vartheta^{\prime}\}$ is infeasible.
This contradicts the hypothesis.

In other words, $\{\mathbf{t}^{*},\vartheta^{*}\}$ is the optimal
solution to problem $\mathscr{P}2$ if $w_{o(k),o(k+1)}^{k}(t_{k}^{*},t_{k+1}^{*})=\vartheta^{*}$
for all $k\in\{1,\cdots,M-1\}$.

\section{Proof of Proposition~\ref{prop:c_increasing_over_vartheta}\label{sec:proof_c_increasing_over_vartheta}}

Denote the time boundaries computed by $\vartheta$ with $\Upsilon(\vartheta;t_{k},t_{k+1})$
as $\{t_{k}(\vartheta)\}$. Next, we will prove the monotonicity of
$\{t_{k}(\vartheta)\}$ over $\vartheta$. Without loss of generality,
we assume $\vartheta_{1}<\vartheta_{2}$, then we will prove the monotonicity
of $t_{k}(\vartheta)$ for $k=2,3,\cdots,M$ by induction. Starting
with $k=2$, we have $w_{o(1),o(2)}^{k}(t_{1},t_{2}(\vartheta_{1}))=\vartheta_{1}<\vartheta_{2}=w_{o(1),o(2)}^{k}(t_{1},t_{2}(\vartheta_{2}))$.
Then, we have $t_{2}(\vartheta_{1})>t_{2}(\vartheta_{2})$ due to
the monotonically decreasing of $w_{m,n}^{k}(t_{k},t_{k+1})$ over
$t_{k+1}$ in Lemma~\ref{lem:monotonicity}.

Assuming $t_{k}(\vartheta_{1})>t_{k}(\vartheta_{2})$ for $k\in\{3,\cdots,M-1\}$,
then 
\begin{align*}
 & w_{o(k),o(k+1)}^{k}\left(t_{k}\left(\vartheta_{1}\right),t_{k+1}\left(\vartheta_{1}\right)\right)=\vartheta_{1}\\
 & \quad<\vartheta_{2}=w_{o(k),o(k+1)}^{k}\left(t_{k}\left(\vartheta_{2}\right),t_{k+1}\left(\vartheta_{2}\right)\right)\\
 & \quad\stackrel{(a)}{<}w_{o(k),o(k+1)}^{k}\left(t_{k}\left(\vartheta_{1}\right),t_{k+1}\left(\vartheta_{2}\right)\right)
\end{align*}
where (a) holds due to $t_{k}(\vartheta_{1})>t_{k}(\vartheta_{2})$
and the monotonically increasing of $w_{m,n}^{k}(t_{k},t_{k+1})$
over $t_{k}$ in Lemma~\ref{lem:monotonicity}. Similarly, because
of the the monotonically decreasing of $w_{m,n}^{k}(t_{k},t_{k+1})$
over $t_{k+1}$ in Lemma~\ref{lem:monotonicity}, we have $t_{k+1}(\vartheta_{1})>t_{k+1}(\vartheta_{2})$.

In conclusion, $t_{k}(\vartheta_{1})>t_{k}(\vartheta_{2})$ for all
$k\in\{2,3,\cdots,M\}$ if $\vartheta_{1}<\vartheta_{2}$. It includes
$t_{M}(\vartheta_{1})>t_{M}(\vartheta_{2})$ if $\vartheta_{1}<\vartheta_{2}$.
In other words, we have $t_{M}(\vartheta)$ is strictly decreasing
over $\vartheta$.

\section{Proof of Proposition~\ref{prop:opt_c_l}\label{sec:proof:opt_c_l}}

If $\Psi(\vartheta_{1})=\varnothing$, it is obvious that $\vartheta(\vartheta_{1})\subseteq\vartheta(\vartheta_{2})$.
If $\Psi(\vartheta_{1})\neq\varnothing$, for any element $\mathbf{L}^{\prime}\in\Psi\left(\vartheta_{2}\right)$,
we have 
\[
\begin{cases}
\int_{t_{k,z}}^{t_{k+1,z}}\mathbb{E}\Big[\log\Big(1+\frac{\vartheta_{1}h_{o\left(k,z\right),o\left(k+1,z\right)}\left(t\right)}{\max_{j\in\mathcal{N}}\left\{ h_{o\left(k,z\right),j}\left(t\right)\right\} \delta^{2}}\Big)\Big]dt\\
\quad\quad\quad\quad\quad\quad\quad\times Bl_{z}^{\prime}\left(t\right)\ge S_{z},\forall z\in\mathcal{Z},\forall k\in\mathcal{M}\\
\boldsymbol{l}^{\prime}\left(t\right)\in\mathcal{L},\forall t
\end{cases}
\]
according to the definition of $\Psi\left(\vartheta\right)$ in (\ref{eq:def_Phi}).
Then there must be 
\[
\begin{cases}
\int_{t_{k,z}}^{t_{k+1,z}}\mathbb{E}\Big[\log\Big(1+\frac{\vartheta_{2}h_{o\left(k,z\right),o\left(k+1,z\right)}\left(t\right)}{\max_{j\in\mathcal{N}}\left\{ h_{o\left(k,z\right),j}\left(t\right)\right\} \delta^{2}}\Big)\Big]dt\\
\quad\quad\quad\quad\quad\quad\quad\times Bl_{z}^{\prime}\left(t\right)\ge S_{z},\forall z\in\mathcal{Z},\forall k\in\mathcal{M}\\
\boldsymbol{l}^{\prime}\left(t\right)\in\mathcal{L},\forall t
\end{cases}
\]
because of the capacity function $\log_{2}(1+\vartheta x)$ is increasing
over $\vartheta$ ($x>0$) and $0\le\vartheta_{1}<\vartheta_{2}$.
Thus $\mathbf{L}^{\prime}$ satisfies conditions (\ref{eq:multi_thp_c_v2})
and (\ref{eq:multi_l_c_v2_2}) when $\vartheta=\vartheta_{2}$. In
other words, $\mathbf{L}^{\prime}\in\Psi\left(\vartheta_{2}\right)$
according to the definition of $\Psi\left(\vartheta\right)$ in (\ref{eq:def_Phi}).
Accordingly, $\vartheta(\vartheta_{1})\subseteq\vartheta(\vartheta_{2})$.

In conclusion, $\vartheta(\vartheta_{1})\subseteq\vartheta(\vartheta_{2})$
if $0\le\vartheta_{1}<\vartheta_{2}$.

\bibliographystyle{IEEEtran}
\bibliography{IEEEabrv,StringDefinitions,JCgroup,ChenBibCV,JCgroup-bw}

\end{document}